\documentclass[12pt,a4paper]{article}
\usepackage[utf8]{inputenc}
\usepackage{amsmath}
\usepackage{amsfonts}
\usepackage{mathtools}
\usepackage{graphicx}
\usepackage{color}
\usepackage[hidelinks]{hyperref}
\usepackage{tabularx}
\usepackage{float}
\usepackage{array}
\usepackage{multirow}
\usepackage{makecell}
\usepackage{url}
\usepackage{colortbl}
\usepackage{dsfont}
\usepackage{stmaryrd}
\usepackage{natbib}
\usepackage{caption}
\usepackage{color}
\usepackage{titlesec}
\renewcommand{\baselinestretch}{1.3}

\makeatletter
\@addtoreset{equation}{section}

\makeatother

\newcounter{Fig}[figure]

\newcounter{Tab}[table]

  {\refstepcounter{Tab}
   \addtocounter{table}{1}
   \samepage\vspace{0.2cm}
   \centerline{#1} \nobreak
   \begin{verse}{Table \thesection.\arabic{Tab}:}
  }%
  {\end{verse} \vspace{0.2cm}}

\relax

\newcommand{\bqa}{\begin{eqnarray*}}
\newcommand{\eqa}{\end{eqnarray*}}
\newcommand{\bqan}{\begin{eqnarray}}
\newcommand{\eqan}{\end{eqnarray}}
\newcommand{\bqt}{\begin{quote}}
\newcommand{\eqt}{\end{quote}}
\newcommand{\bt}{\begin{tabbing}}
\newcommand{\et}{\end{tabbing}}
\newcommand{\bit}{\begin{itemize}}
\newcommand{\eit}{\end{itemize}}
\newcommand{\ben}{\begin{enumerate}}
\newcommand{\een}{\end{enumerate}}
\newcommand{\beq}{\begin{equation}}
\newcommand{\eeq}{\end{equation}}
\newcommand{\bdefi}{\begin{definition}}
\newcommand{\edefi}{\end{definition}}
\newcommand{\bpro}{\begin{proposition}}
\newcommand{\epro}{\end{proposition}}
\newcommand{\blem}{\begin{lemma}}
\newcommand{\elem}{\end{lemma}}
\newcommand{\bth}{\begin{theorem}}
\newcommand{\eth}{\end{theorem}}
\newcommand{\bco}{\begin{corollary}}
\newcommand{\eco}{\end{corollary}}
\newcommand{\bdes}{\begin{description}}
\newcommand{\edes}{\end{description}}
\newcommand{\bre}{\begin{remark}}
\newcommand{\ere}{\end{remark}}

\newtheorem{definition}{Definition}[section]
\newtheorem{proposition}[definition]{Proposition}
\newtheorem{lemma}[definition]{Lemma}
\newtheorem{theorem}[definition]{Theorem}
\newtheorem{corollary}[definition]{Corollary}
\newtheorem{remark}[definition]{Remark}

\newenvironment{proof}[1][Proof]{\textbf{#1.} }{\ \rule{0.5em}{0.5em}}

\begin{document}

\begin{titlepage}

\title{Parametric insurance for extreme  risks: the challenge of properly covering severe claims}

%Estimation in Single Index Regression Models \\ with Censored Responses}

\author{{\large Olivier L\textsc{opez}$^1$},
{\large Maud T\textsc{homas}$^1$}  }

\date{\today}
\maketitle

\renewcommand{\baselinestretch}{1.1}

\begin{abstract}
 Parametric insurance has emerged as a practical way to cover risks that may be difficult to assess. By introducing a parameter that triggers compensation and allows the insurer to determine a payment without estimating the actual loss, these products simplify the compensation process, and provide easily traceable indicators to perform risk management. On the other hand, this parameter may sometimes deviate from its intended purpose, and may not always accurately represent the basic risk. In this paper, we provide theoretical results that investigate the behavior of parametric insurance products when faced with large claims. In particular, these results measure the difference between the actual loss and the parameter in a generic situation, with a particular focus on heavy-tailed losses. These results may help to anticipate, in presence of heavy-tail phenomena, how parametric products should be supplemented by additional compensation mechanisms in case of large claims. Simulation studies, that complement the analysis, show the importance of nonlinear dependence measures in providing a good protection over the whole distribution.
\end{abstract}

\vspace*{0.5cm}


\vspace*{0.5cm}

\noindent{\bf Short title:} Parametric insurance and extreme risks.

\vspace*{.5cm}

{\small
\parindent 0cm
$^1$ Sorbonne Universit\'e, CNRS, Laboratoire de Probabilit\'es, Statistique et Mod\'elisation, LPSM, 4 place Jussieu, F-75005 Paris, France}

\end{titlepage}

\small
\normalsize
\addtocounter{page}{1}

\section{Introduction}

Parametric insurance is a very elegant and efficient way to simplify risk management in situations where the assessment of losses might be difficult \citep[see e.g.][]{lin2020application}. Parametric solutions have in particular been developed in the case of natural disasters \citep[see e.g.][]{van2011parametric,horton2018parametric}. A typical example is the case of an hurricane striking some area. The damages of such an episode can be very difficult to assess, leading to expertise costs and delays in the compensation process. The solution proposed by parametric insurance is not to cover  directly actual losses, but to work with some ``parameter", that is a quantity related to the loss and easily measurable. In the case of natural disasters, wind speed, precipitation level, or any index based on relevant physical quantities can be used. \citet{figueiredo2018probabilistic} described a detailed methodology in the example of parametric flood insurance in Jamaica. Since the parameter can be measured instantly (or in a short period of time), payment can be made more quickly. Furthermore, when it comes to assessing the risk, the situation is greatly simplified if one is working with an readily available quantity that can be tracked and modeled through standard actuarial methods. This explains the growing popularity of these solutions that are now widely used \citep[see e.g.][]{prokopchuk2020parametric,broberg2020parametric, bodily2021portfolio}. Moreover, parametric products open the way to insurance linked securities, with the example of Catbonds \citep[see e.g.][]{jarrow2010simple,albrecher2004qmc, zimbidis2007modeling}.

Nevertheless, parametric insurance is not a miracle solution. Reducing the volatility of the outcome has a cost. One of the difficulties is to convince the policyholder  of the relevance of a guarantee based on a given parameter. The attractiveness of such contracts may be reduced by the fear of the customers that the policy does not adequately cover the very risk they want to be protected against. This is particularly true if the parameter is complex and may not be fully understandable. In this case, simplifying the compensation process is not necessarily worth the loss in terms of protection. In addition, there are miscalculations, such as those reported by \citet{johnson2021paying}, that can discourage policyholders. \citet{johnson2021paying} details some ``Ex gratia'' payments that are sometimes activated to limit the impact of these inconveniences.

The aim of this article is to study, in a general simplified framework, under which conditions parametric insurance may still provide (or not) a good protection against  risk in case of large claims, taking the point of view of the policyholder. By large claims, we mean that the actual loss of the policyholder is large. These situations, which deviate from the central scenario that is supposed to guide the calibration of the parameter-based payoff, require special attention because they correspond to situations that may not be the most likely, but that correspond to important policyholder concerns: if the potential customers feel that they are not adequately covered in case of severe events---which is at the heart of the decision to use insurance---they may be reluctant to purchase such solutions.

In this work, we focus on two special cases. In the first case, the loss variable is supposed to have a Gaussian tail. In this situation, significant deviations from the central scenario are of low probability. Therefore, simply working on the correlation between the parameter and the loss is sufficient to improve the risk coverage. However, in the second case, losses with heavy tail are more challenging. Our results show that extreme losses may be difficult to capture, unless the parameter is able to reflect this heaviness. These results are consistent with practical observations from \cite{johnson2021paying}, and can be seen as a first step to improving parametric products by anticipating the disappointment of some policyholders in case of very large claims. We illustrate these properties with a simulation study inspired by the case of cyber risk, specifically data breaches insurance. In this case, the volume of data that are breached can be related to an estimated cost, and using this indicator to design parametric insurance products would make sense. The results are also illustrated on a real data set for extreme flood events. 

The rest of the paper is organized as follows. In Section \ref{sec_1}, we discuss the general framework we consider for evaluating the performance of a parameter used in parametric insurance. In particular, we focus on the question of the dependence between the parameter and the actual loss, which is essential to expect to obtain satisfactory properties. Section \ref{sec_2} gathers some theoretical results to measure how the parametric solution is able to approximate the true loss when the amount of the latter is high. The simulation study illustrating these properties in the case of cyber insurance is described in Section \ref{sec_cyber}. Finally, we provide a real data example on extreme flood events in Section \ref{sec:floods}. The proof of the technical results are listed in the appendix in Section \ref{sec_app}.

\section{Parametric insurance model}
\label{sec_1}

In this section, we explain the general framework used to model the difference between the actual loss and the payment made via the parametric insurance product. The general framework is described in Section \ref{sec_desc}. The key issue of the dependence between the two variables (actual loss and payment) is addressed in Section \ref{sec_dep} which introduces tools from copula theory.

\subsection{Description of the framework}
\label{sec_desc}

In the following, we consider a random variable $X$ representing the actual loss incurred by a policyholder. Parametric insurance relies on the fact that $X$ may be difficult to measure. In case of a natural disaster, $X$ may be, for example, the total cost of the event on a given area. It may take time to accurately estimate this cost (if even possible), and the idea is rather to pay a cost $Y$ that is not exactly $X,$ but is supposed to reflect it. $Y$ is supposed to be a variable that should be easier to measure.

For example, the precipitation level $\theta$, or other weather variables, can be obtained instantaneously, and a payoff can be derived from $Y,$ that is, in this case, $Y=\phi(\theta)$ for a given non-decreasing function $\phi.$ We will use the term ``parameter'' to refer the random variable $\theta.$

Ideally, we would like $Y$ to be close to $X.$ Another advantage of this approach is the potential reduction of volatility: paying $Y$ instead of $X$ is interesting in terms of risk management if the variance $\sigma_Y^2$ of $Y$ is small compared to the variance $\sigma^2_X$ of $X.$ Of course, if the variance of $Y$ is too small compared to $X,$ the quality of the approximation of $X$ by $Y$ may decrease, since the distribution of $Y$ does not match that of $X.$

\subsection{Dependence}
\label{sec_dep}

For the parameter $\theta,$ or, more precisely, the payoff $Y=\phi(\theta),$ to  accurately describe the risk, $X$ and $Y$ must be dependent. The simplest way to describe this dependence is by correlation. More precisely, let $\rho$ be the correlation coefficient of $X$ and $Y$ defined by  $\rho=\mathrm{Cor}(X,Y)=\mathrm{Cov}(X,Y)\sigma_X^{-1}\sigma_Y^{-1},$ where $\sigma_X^2=\mathrm{Var}[X]$ and $\sigma_Y^2=\mathrm{Var}[Y]$. Considering the quadratic loss, we have
$$\mathbb{E}[(X-Y)^2]=\sigma_X^2+\sigma_Y^2-2\rho \sigma^2_X\sigma_Y^2+(\mathbb{E}[X^2]+\mathbb{E}[Y^2]-2\mathbb{E}[X]\mathbb{E}[Y]).$$
Hence, increasing this correlation reduces the loss.

However, correlation is known to be a measure of dependence that primarily considers 
the center of the distribution, but not the tail. When faced with a large claim, that is when $X$ is far from its expectation, correlation is not enough to ensure that $Y$ remains close to its target.

To illustrate this issue, consider the case of covering claims where $X$ exceeds a deductible $x_0.$ If $X$ is not observed and the insurance product is based on the parameter $\theta,$ the insurance company will make some mistakes: sometimes a payout will be initiated when $X<x_0,$ and sometimes no payout will occur even if $X\geq x_0.$ This is because $\theta$ is just a proxy to get to $X$: the payment is actually triggered when $\theta\geq t_0,$ and the event $\{\theta\geq t_0\}$ does not exactly match the event $\{X\geq x_0\}.$ A good parameter must be such that $\pi_+(t_0,x_0)=\mathbb{P}(\theta\geq t_0|X\geq x_0)$ and $\pi_{-}(t_0,x_0)=\mathbb{P}(\theta < t_0|X<x_0)$ are both close to 1. Maximizing $\pi_+$ is supposed to improve the satisfaction of the policyholder: coverage is obtained for almost all claims that are significant. On the other hand, a high value of $\pi_-$ guarantees that the insurer will not pay for claims that were initially beyond the scope of the product.

Let us introduce the cumulative distribution function (c.d.f.) $F_{\theta}(t)=\mathbb{P}(T\leq t)$ (resp. $F_X(x)=\mathbb{P}(X\leq x)$) defining the distribution of $\theta$ (resp. of $X$), and the joint c.d.f. $F_{\theta,X}(t,x)=\mathbb{P}(\theta\leq t,X\leq x).$ A common and general way to describe the dependence structure between $\theta$ and $X$ is to use copulas. The copula theory is based on the fundamental result of \citet{sklar} which states that
\begin{equation}F_{\theta,X}(t,x) = \mathfrak{C}(F_{\theta}(t),F_{\theta}(x)), \label{cop}
\end{equation}
where $\mathfrak{C}$ is a copula function, that is, the joint distribution function of a two-dimensional variable on $[0,1]^2$ with uniformly distributed margins. The decomposition is unique if $\theta$ and $X$ are continuous, which is the assumption we make in the following. Hence, (\ref{cop}) shows that there is a separation between the marginal behavior of $(\theta,X),$ and the dependence structure contained in $\mathfrak{C}.$ Many parametric families of copulas have been proposed to describe various forms of dependence   \citep[see e.g.][]{nelsen}.

Let us write $\pi_+$ and $\pi_-$ in terms of copulas. Introducing the survival functions $S_{\theta}(t)=1-F_{\theta}(t),$ $S_X(x)=1-F_X(x),$ and $S(t,x)=\mathbb{P}(\theta\geq t,X\geq x),$ we have
\begin{eqnarray*}
\pi_+(t_0,x_0) &=& \frac{\mathfrak{C}^*(S_{\theta}(t_0),S_{X}(x_0))}{S_{X}(x_0)}, \\
\pi_-(t_0,x_0) &=& \frac{S_{\theta}(t_0)-S(t_0,x_0)}{F_X(x_0)},
\end{eqnarray*}
where $\mathfrak{C}^*$ is the survival copula associated with $\mathfrak{C},$ that is
$$\mathfrak{C}^*(v,w)=v+w-1+\mathfrak{C}(1-x,1-w).$$

 To relate this to classical dependence measures for $\pi_+$, consider the case where $S_{\theta}(t_0)=S_X(x_0)=u.$ In this situation, a large value of $\pi_+$ means that the deductible on $\theta$ that we use (based on a quantile of the distribution of $\theta$), has approximately the same effect as a deductible directly on $X$ (based on the same quantile). $\pi_+$ close to 1 yields
$\mathfrak{C}^*(u,u)/u \approx 1$. If $u$ is small (which means that we are focusing on higher quantiles, that is large claims), $\pi_+$ becomes close to $$\lambda=\lim_{u\rightarrow 0}\frac{\mathfrak{C}^*(u,u)}{u}=\lim_{u\rightarrow 0}\mathbb{P}(\theta \geq S_\theta^{-1}(u)|X\geq S_X^{-1}(u)),$$ which is the upper tail dependence  \citep[see][]{nelsen}. This shows that if we focus on large claims, correlation is not sufficient  to correctly represent the risk, tail dependence seems more relevant.

In this discussion, we only have  focused on what triggers the payment, that is, when $\theta>t_0.$ However, the difference between the payment $Y=\phi(\theta)$ and the actual loss $X$ must also be studied, which is the purpose of the next section. 
\begin{remark}
$\pi_-$ can also be expressed in terms of a copula, but is more difficult to link to the classical dependence measure. Our scope being essentially to focus on the tail of the distribution and on the potential difference between what the customer expects and what the customers gets, we do not develop this point.
\end{remark}

\section{Difference between the actual loss and the payoff based on the parameter}
\label{sec_2}

In this section, we provide theoretical results to quantify the difference between $X$ and $Y$ when a claim is large, that is when $X$ is large. The quantities measured are defined in Section \ref{sec_mes}. We then consider two types of distribution: Gaussian variables (Section \ref{sec_gauss}) are used as a benchmark, while heavy-tailed variables are considered in Section \ref{sec_heavy}.

\subsection{Measure the difference}
\label{sec_mes}

In the following, we consider two different quantities to measure the distance between $Y=\phi(\theta)$ and $X$ for large claims, that is when $X$ exceeds some high value $s.$

The first measure is $\mathbb{E}[X-Y|X\geq s].$ The advantage of this measure is that it shows if $Y$ tends to be smaller or larger than $X.$ On the other hand, the conditional bias $\mathbb{E}[X-Y|X\geq s]$ can be zero (if $\mathbb{E}[Y|X]=X$) while the conditional variance can be large, leading to potentially huge gaps between $X$ and $Y$ in practice.

For this reason, we also consider a classical quadratic loss, that is
$\mathbb{E}[(X-Y)^2|X\geq s].$ Note that this quadratic loss may not be defined for distributions that have a too heavy tail (this is also the case for $\mathbb{E}[X-Y|X\geq s]$ which may not be defined if the expectation is infinite, but the assumption of a finite variance further restrains the set of distributions).

To understand how the approximation made by the parametric approach deteriorates when $X$ is large, we will next derive asymptotic approximations of these quantities when $s$ tends to infinity.

\subsection{Gaussian losses}
\label{sec_gauss}

In this section, we assume that $(X,Y)$ are Gaussian variables with 
\begin{equation}
    \label{loi_gauss}
\left(\begin{array}{c} X \\ Y \end{array}\right)\sim \mathcal{N}\left(\left(\begin{array}{c}\mu_X \\ \mu_Y \end{array}\right),\left(\begin{array}{cc}\sigma^2_X & \rho \sigma_X\sigma_Y \\ \rho \sigma_X\sigma_Y & \sigma_Y^2\\
\end{array}\right)\right).\end{equation}

Considering such a framework is not entirely realistic, in the sense that $X$ and $Y$ could take negative values with a non-zero probability. Nevertheless, if $\mu_X$ and $\mu_Y$ are large enough, the probability of such an event is quite small. The Gaussian case is considered here mainly because it gives us an example of variables strongly concentrated around their expectation, in order to measure the difference with heavy-tailed variables of Section \ref{sec_heavy}.

In addition, another motivation for considering Gaussian variables is the Central Limit Theorem. If $X$ consists of the aggregation of individual claims, that is $X=\sum_{i=1}^n Z_i,$ where $(Z_i)_{1\leq i \leq n}$ are independent identically distributed losses, the Central Limit Theorem states that $X$ is approximately distributed as a Gaussian random variable with mean $n \mathbb{E}[Z_1]$ and variance $n^{1/2}\mathrm{Var}(Z_1),$ provided that $n$ is sufficiently large. A Gaussian limit can also be obtained under certain weak forms of dependence for these aggregated losses. This requires of course that the variance of $Z_1$ is finite.

A specificity of Gaussian random vectors is that their dependence structure is solely determined by the correlation matrix. Here, the dependence is driven by the correlation coefficient $\rho.$ As already mentioned, this is somehow a way to define dependence in the central part of the distribution. Because of the particular structure of Gaussian variable, this quantity has also an effect on the tail, that is we consider even situations where $X\geq s$ with $s$ large.

Propositions \ref{prop_gauss1} and \ref{prop_gauss2} provide explicit formulas for $\mathbb{E}[X-Y|X\geq s]$ and $\mathbb{E}[(X-Y)^2|X\geq s].$

\begin{proposition}
\label{prop_gauss1} Consider a random vector distributed as (\ref{loi_gauss}). Then, as $s \to +\infty$,
\begin{equation}\label{g1}\mathbb{E}[X-Y|X\geq s]\sim (\mu_X-\mu_Y)+\left(1-\frac{\rho \sigma_Y}{\sigma_X}\right)(s-\mu_X),
\end{equation}
where $\sim$ is the symbol for equivalence.
\end{proposition}

Note that if $\mathbb{E}[Y|X]=X,$ that is if $\rho \sigma_Y\sigma_X^{-1}=1$ and $\mu_X=\mu_Y,$ we find that $\mathbb{E}[X-Y|X\geq s]=0.$ Apart from this trivial case, we can decompose (\ref{g1}) into two parts. First, the difference between the expectations of $\mu_X$ and $\mu_Y,$ which reflects the ability of $Y$ to capture the central part of the distribution of $X.$ The second term increases with $s$ when $\rho \sigma_Y\sigma_X^{-1}<1.$ However, a large value of the correlation coefficient $\rho$ tends to reduce this effect. Therefore, we can rely on  a strong correlation between $X$ and $Y$ to improve the ability of the parametric insurance contract to provide good results even for large claims.

On the other hand, the case $\rho \sigma_Y\sigma_X^{-1}\geq 1$ is less interesting to study: it would correspond to a situation where $\sigma_Y\geq \sigma_X,$ that is a payoff based on the parameter which is more volatile that the one we would have directly used $X.$ While this situation may occur, it is not the ideal case where parametric insurance is used to both facilitate the collection of information required to trigger claim payment, and reduce the uncertainty.

Similar observations apply from the Proposition \ref{prop_gauss2} result for the quadratic loss.

\begin{proposition}
\label{prop_gauss2} Consider a random vector distributed as (\ref{loi_gauss}). Then, as $s \to +\infty$,
\begin{equation}\label{g2}\mathbb{E}[(X-Y)^2|X\geq s]\sim \left(1-\rho\frac{\sigma_Y}{\sigma_X}\right)^2 \frac{s^2}{\sigma_X^2}.\end{equation}
\end{proposition}

The exact value for $\mathbb{E}[(X-Y)^2|X\geq s]$ can also be computed for a vector distributed as (\ref{loi_gauss}). The formula can be obtained from the proof of Proposition \ref{prop_gauss2}, which is made in Section \ref{sec_preuve_gauss2}.

\subsection{Heavy-tailed losses}
\label{sec_heavy}

Heavy-tailed random variables play an important role in Extreme Value Theory, \citep[see e.g.][]{beirlant, coles2001introduction}. Suppose that we are dealing with an i.i.d sample $Y_1,\ldots, Y_n$ whose c.d.f. $F$ satisfies the following property 
\begin{equation}\label{rv}
F(t) = t^{-\gamma}\ell(t)
\end{equation}
where $\ell$ is a slowly-varying function, namely 
$$\forall x>0, \; \lim_{t\rightarrow \infty} \frac{\ell(tx)}{\ell(t)}=1.$$ The fundamental result of Extreme Value Theory states that the excesses above a certain threshold $u$ of the sample, defined as $Y_i -u$ given that $Y_i >u$, converges in distribution toward a non-degenerated distribution and that this distribution necessarily  belongs  to a parametric family of distributions, called the generalized Pareto distributions. More precisely,  \citet{balkema1974residual} has shown, that under the hypothesis \eqref{rv} when $\gamma>0$, there exists normalization  constants $a_u$ and $b_u>0$ such that 
$$\mathbb{P}\left(Y_i -u\geq a_u +b_ux \mid Y_i >u\right)\underset{u \to \infty}{\longrightarrow} H_\gamma(x),$$
and $H_\gamma(x)$ is necessarily of the form, for $x>0$,
$$
H_\gamma(x)=1-(1+\gamma x)^{1/\gamma}.
$$
Furthermore, \citet{gnedenko} showed that the survival function $1-F(t)$ is necessarily of the form \eqref{rv}. The parameter $\gamma$ which reflects the heaviness of tail of $F$ is called the tail index. The higher $\gamma,$ the heavier the tail of the distribution: $X$ tends to take large values with a significant probability. Here, we consider the context where $\gamma>0$ and $X$ belongs to the heavy tail domain. The right-tail of heavy tail distributions decrease polynomially toward 0, their moments of order larger than $1/\gamma$ do not exist. For example, the Pareto, the Student, the log-gamma and the Cauchy distributions are heavy-tailed.

Hence, Assumption \ref{rv} allows us to cover a large set of distributions. In the following, we thus assume that
\begin{equation}\label{assum:heavy}
S_X(t) =\ell_X(t)t^{-1/\gamma_X},  \quad \text{ and } \quad S_Y(t) =\ell_Y(t)t^{-1/\gamma_Y}, 
\end{equation}
with $\gamma_X,\gamma_Y >0$ and $\ell_X$ and $\ell_Y$ two slowly-varying functions.

Note that heavy-tailed variables can also be used to approximate sums of losses. Considering the example of $X=\sum_{i=1}^n Z_i,$ if $Z_i$ are heavy-tailed i.i.d. random variables, \citet{mikosch} show that, when dealing high quantiles, $X$ can be approximated by a heavy-tailed variable.

We do not provide here explicit values for $\mathbb{E}[X-Y|X\geq s]$ and $\mathbb{E}[(X-Y)^2|X\geq s],$ since these quantities depend on $\ell_X$ and $\ell_Y.$ Nevertheless, our results should provide general bounds for large values of $s$. We first consider the case of $\mathbb{E}[X-Y|X\geq s]$ in Proposition \ref{prop1}. We recall that $\mathbb{E}[X-Y|X\geq s]$ (resp. $\mathbb{E}[(X-Y)^2|X\geq s]$) is defined only if $\gamma_X$ and $\gamma_Y$ are less than 1 (resp. less than 0.5).

\begin{proposition}
\label{prop1}
Consider $X\geq 0$ and $Y\geq 0$ with survival functions as in \eqref{assum:heavy}, with $\gamma_X>\gamma_Y.$ There exists a constant $c>0$ depending on $\ell_X$ and $\ell_Y$ and not on their dependence structure, such that, for $s$ large enough,
$$\mathbb{E}[X-Y|X\geq s]\geq cs.$$
\end{proposition}

Proposition \ref{prop1} shows that there is a linear increase in this difference for large values of $s.$ However, the situation is quite different from the Gaussian case. In the Gaussian case, we might expect to reduce the slope by relying on a strong correlation between $X$ and $Y.$ Here, improving the correlation would certainly have an effect, but not on the leading linear term. In fact, the distribution of $X$ being heavier than $Y,$ when $s$ becomes large, $X$ belongs to some areas which are inaccessible by $Y$ except with a very low probability. In practice, the difference between $\gamma_Y$ and $\gamma_X$ also plays a role, but again only on terms of smaller order.

We can obtain a more precise result under some additional assumptions, as we see in Proposition \ref{prop2} below.

\begin{proposition}
\label{prop2}
Consider $X\geq 0$ and $Y\geq 0$ with survival functions as in \eqref{assum:heavy} with $\gamma_X>\gamma_Y.$ Moreover, let
$$\psi(x)=\mathbb{E}[Y|X=x].$$
Assume that $\psi$ is strictly non decreasing and such that the random variable $\psi(X)$ is heavy-tailed. Then%that $\psi^{-1}(t)=\ell_{\psi}(t)t^{-1/\gamma_{\psi}}$ where $\ell_{\psi}$ is a slow-varying function. Then
$$\mathbb{E}[X-Y|X\geq s]=s-\psi(s)+o(s).$$
\end{proposition}

Under this additional assumption, we see that we even get a linear increase of $\mathbb{E}[X-Y|X\geq s],$ since $\psi(s)$ is less than $s$ since $\gamma_Y<\gamma_X.$

A similar result is obtained for the quadratic loss in Proposition \ref{prop3}, where we see that this quantity increases at rate $s^2,$ regardless the dependence structure between $X$ and $Y.$

\begin{proposition}
\label{prop3}
Consider $X\geq 0$ and $Y\geq 0$ with survival functions as in \eqref{assum:heavy}, with $\gamma_X>\gamma_Y.$
$$\mathbb{E}[(X-Y)^2|X\geq s]=s^2+o(s^2).$$
\end{proposition}

From these theoretical results, we see that there is generally a large gap between the payoff $Y$ and the actual loss $X$ when the variables are heavy-tailed. We have assumed here that $\gamma_X>\gamma_Y,$ which seems to be the most interesting case since, in this situation, the risk taken by parametric insurance is less volatile than the initial risk. In the opposite situation, the parametric product would tend to overcompensate the actual loss. The $\gamma_X=\gamma_Y$ situation seems purely theoretical, as it would require a very special situation in which the two tail indices are exactly the same.

Finally, let us mention that all the results of this section extend to the case where $X$ is heavy-tailed and $Y$ is light-tailed. In this situation, the remaining terms in the asymptotic expansions are even smaller.

\section{Simulation study from the cyber risk context}
\label{sec_cyber}

The purpose of this section is to illustrate the theoretical results and to go beyond the asymptotic approximations we have derived. Simulation settings are inspired from issues in the cyber insurance domain. Cyber insurance is an area where parametric insurance is increasingly mentioned as a promising tool to design contracts adapted to the complexity of the risk \citep[see e.g.][Chapter 5]{dal2020towards,occ}. Various indicators have been proposed to track risk, whether it is the frequency or severity of an event. Here we build a simulation setting inspired (in terms of the distribution used) by calibrations performed in the cyber domain, more precisely in the case of data breaches. This context offers a natural physical parameter which describes the severity. The results that we proposed are based on previous work of the volume of leaked data and its link with the cost of an event, and can be understood as a sensitivity analysis of the models that are considered. { In Section \ref{sec_number_records}, we give a brief presentation of this context. The simulation settings we consider are described in Section \ref{sec_illustration}, with a focus on the copula models used in Section \ref{sec_copulas}. The results and analyses are presented in Section \ref{sec_results}}.

\subsection{Description of data breaches through number of records}
\label{sec_number_records}

Of the various types of situations that underlie the concept of cyber risk, data breaches are probably the ones for which the cost associated with such an event is relatively easy to assess. Indeed, the volume of data that has been breached (i.e. the ``number of records'') is a good indication of the severity. This quantity, which can be easily measured soon after a claim occurs, can be used as a parameter that should be able to give some indications on the actual loss.

\citet{cost_jacobs_formula} proposed a relationship between this number of records, say $\theta,$ and $X,$ which can be seen as the formula defining the payoff $Y.$ This relationship is of the following type
\begin{equation}
    \label{payoff}
\log Y=\alpha + \beta \log \theta,
\end{equation}
where $X$ and $Y$ are expressed in dollars. The formula was calibrated using data from the Ponemon institute \citep[][]{CODB_2018}. \citet{cost_jacobs_formula} estimated $\alpha=7.68$ and $\beta=0.76.$ Nonetheless, \citet{farkas2021cyber} pointed that the formula, computed from data collected in 2014, was inconsistent with some of the so-called ``mega-breaches'' observed afterwards. For example, two mega-breaches were reported in the 2018 CODB report  \citep[][]{CODB_2018}. The first, with 1 million records breached, would result in an estimated cost of nearly \$79 million, while the actual cost was approximately \$39 million. The largest one, with 50 million records breached, would lead to $Y=1.5$ billion dollars, far from the  \$350 million paid out. This is why \citet{farkas2021cyber} proposed a (very rough) recalibration of the parameters, taking $\alpha=9.59$ and $\beta=0.57.$

Behind this discussion, we see that, although we are dealing with an indicator (the number of records) that seems physically relevant to describe the magnitude of the event under consideration, the payoff function $Y$ may be a very rough approximation of the actual loss, especially in the tail. In addition to the variance inherent in the calibration error of \eqref{payoff}, and the potential lack of fit of the model, we see that there are many uncertainties regarding the estimation of the parameters.

The examples we use in the following are inspired by this example, and the corresponding parameter values.

\subsection{Difference between the basis risk and the parametric coverage}
\label{sec_illustration}

In this section, we illustrate the theoretical results of Propositions \ref{prop1} and \ref{prop3} on simulations inspired by the relationship (\ref{payoff}) between the expected cost of a data breach and the number of records. Different dependence structures are considered.

\paragraph{Simulation setting} In this paragraph, we describe our main simulation setting.
First, in order to consider reasonable values for our example, we need an appropriate distribution for the parameter $\theta.$ We choose to simulate $\theta$ according a simple Pareto distribution with parameters $b=0.7$ and $u=7.10^4$, as considered by \citet{maillart}, namely
$$\mathbb{P}(\theta\geq t)=\left(\frac{u}{t}\right)^{b}, \; \mathrm{for} \; t\geq u.$$
 This choice of  heavy-tailed distributions is consistent with the work of \citet{maillart}, and with the analysis of \citet{farkas2021cyber} based on more recent data (a significant class of the data breaches has been identified to follow a distribution close to that considered by \citet{maillart}).

Then, if the payoff $Y$ is defined according to \eqref{payoff}, $Y$ inherits the heavy-tail property of $\theta$. More precisely, if $\theta$ is distributed as Pareto distribution with parameters $u$ and $b$, then $Y$ is distributed according to a Pareto distribution with parameters $u'=u\exp(\alpha)$ and $b' = b/\beta$. 
% $$\mathbb{P}\left(Y\geq t\right)=\left(\frac{u'}{t}\right)^{b/\beta}, \; \mathrm{for} \; t\geq u',$$ with $u'=u^{\beta}\exp(\alpha).$ 
The parameter $\beta$ is of course the most important when focusing on the tail of the distribution, since it is directly linked to the tail index $\gamma_Y=\beta/b.$ In this work, the coefficients are chosen $\alpha=9.59$ and $\beta=0.5,$ that is slightly lower than the parameter $\beta$ calibrated in \citep{farkas2021cyber}. 

{Now, to simulate the actual loss $X$, we consider that $X$ is linked to the volume of breached data, yet we do not have access to the actual number of records. We thus suppose that the actual loss $X$ also satisfies the relation \eqref{payoff} but with a different parameter $\theta'$ and different coefficients $\alpha'$ and $\beta'$. The parameter $\theta'$ is also simulated according to a Pareto distribution with parameters $u$ and $b'-0.1$. The choice of $b'-0.1$ is motivated by the fact that it corresponds to the margin of error given by \citet{maillart}. Then, the dependence structure between $X$ and $Y$ is created through the dependence between $\theta$ and $\theta'$ described by a copula function. The different copula families and how the corresponding parameters were chosen are explained in the next paragraph.}

% \textbf{Then, to simulate $X,$ we consider
% $$\log X=\alpha+\beta'\log (\theta'),$$
% where
% $$\mathbb{P}(\theta\geq t)=\left(\frac{u}{t}\right)^{b'}, \; \mathrm{for} \; t\geq u,$$
% with $b'=b-0.1.$ }

\paragraph{Copula families}
\label{sec_copulas}

We consider three families of copulas, corresponding to different types of dependence structure. The functions of copula are summarized in Table \ref{tab_copula}.

\begin{table}[!h]
\footnotesize
    \centering
    \begin{tabular}{c|c|c|c|c}
         Copula family & Copula function & $\delta\in$ & $\tau$ & $\lambda_U$ \\
         \hline
         Clayton survival & $u+v-1+\left[\max\left(u^{-\delta}+v^{-\delta}-1,0\right)\right]^{-1/\delta}$ & $\delta\geq -1$ and $\delta\neq 0$ &  $\frac{\delta}{\delta+2}$ & $2^{-\delta}$\\
         \hline
        Gumbel & $\exp\left(-\left[\left(-\log(- u)\right)^{\delta}-\left(\log(- v)\right)^{\delta}\right]^{1/\delta}\right)$ & $\delta\geq 1$ & $\frac{\delta-1}{\delta}$ & $2-2^{-\delta}$ \\
        \hline
        Frank & $-\frac{1}{\delta}\log\left(1+\frac{(e^{-\delta u}-1)(e^{-\delta v}-1)}{e^{-\delta}-1}\right)$ & $\delta\neq 0$ & Implicit & 0 \\
    \end{tabular}
    \caption{Different copula families for the dependence structure of $(\theta, \theta')$.}
    \label{tab_copula}
\end{table}

The Frank copula family is classical and flexible, but does not allow for tail dependence. In contrast, the other two considered families, that is the Gumbel and the Clayton survival copulas, have non-zero tail dependence.

To make things comparable, we consider values of parameters that provide similar dependence structure according to an appropriate indicator. The dependence measure we use is Kendall's tau coefficient $\tau$, which is defined, for two random variables $(\theta,\theta'),$ as follows
$$\tau=\mathbb{P}((\theta_1-\theta_2)(\theta'_1-\theta'_2)>0)-\mathbb{P}((\theta_1-\theta_2)(\theta'_1-\theta'_2)<0),$$
where $(\theta_1,\theta'_1)$ and $(\theta_2,\theta'_2)$ are independent copies of $(\theta,\theta').$ This gives a simple relation between the parameter $\tau$ and the classical parametrization of the copula families.

We consider three values of the parameters for each copula family, corresponding  to $\tau=0.3,$ $\tau=0.5$ and $\tau=0.7$, respectively.

\paragraph{Simulation results}
\label{sec_results}

Figure \ref{fig1} shows the evolution of $\mathbb{E}[(X-Y)|X\geq s]$ and $\mathbb{E}[(X-Y)^2|X\geq s]$ for the different copula families with different dependence structure. Let us recall that these models only  differ  in the dependence structure between $Y$ and $X.$ In each case, $X$ and $Y$ are heavy-tailed.

\begin{figure}
    \centering
    \begin{tabular}{cc}
     \includegraphics[scale = 0.45]{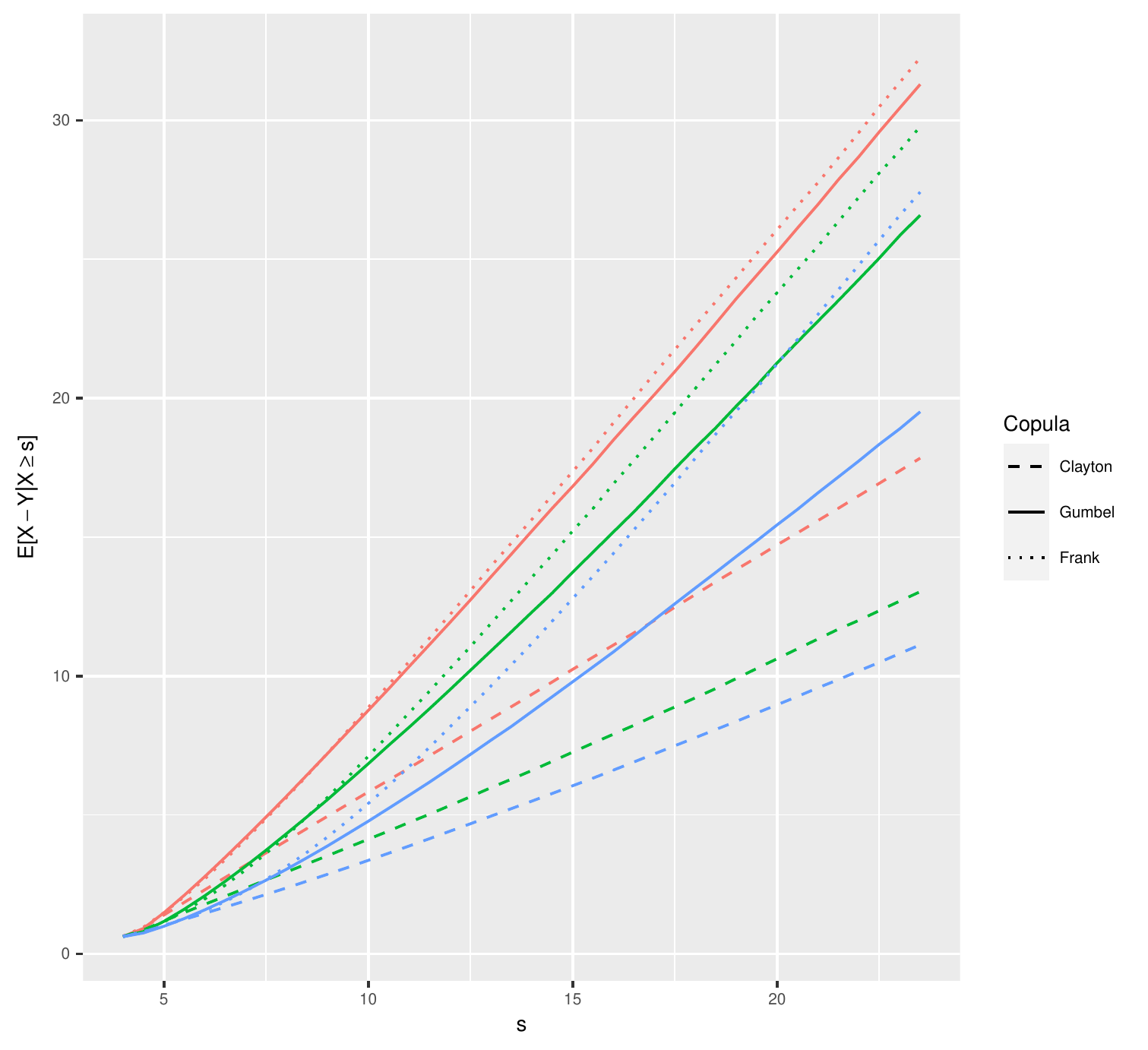} &  \includegraphics[scale = 0.45]{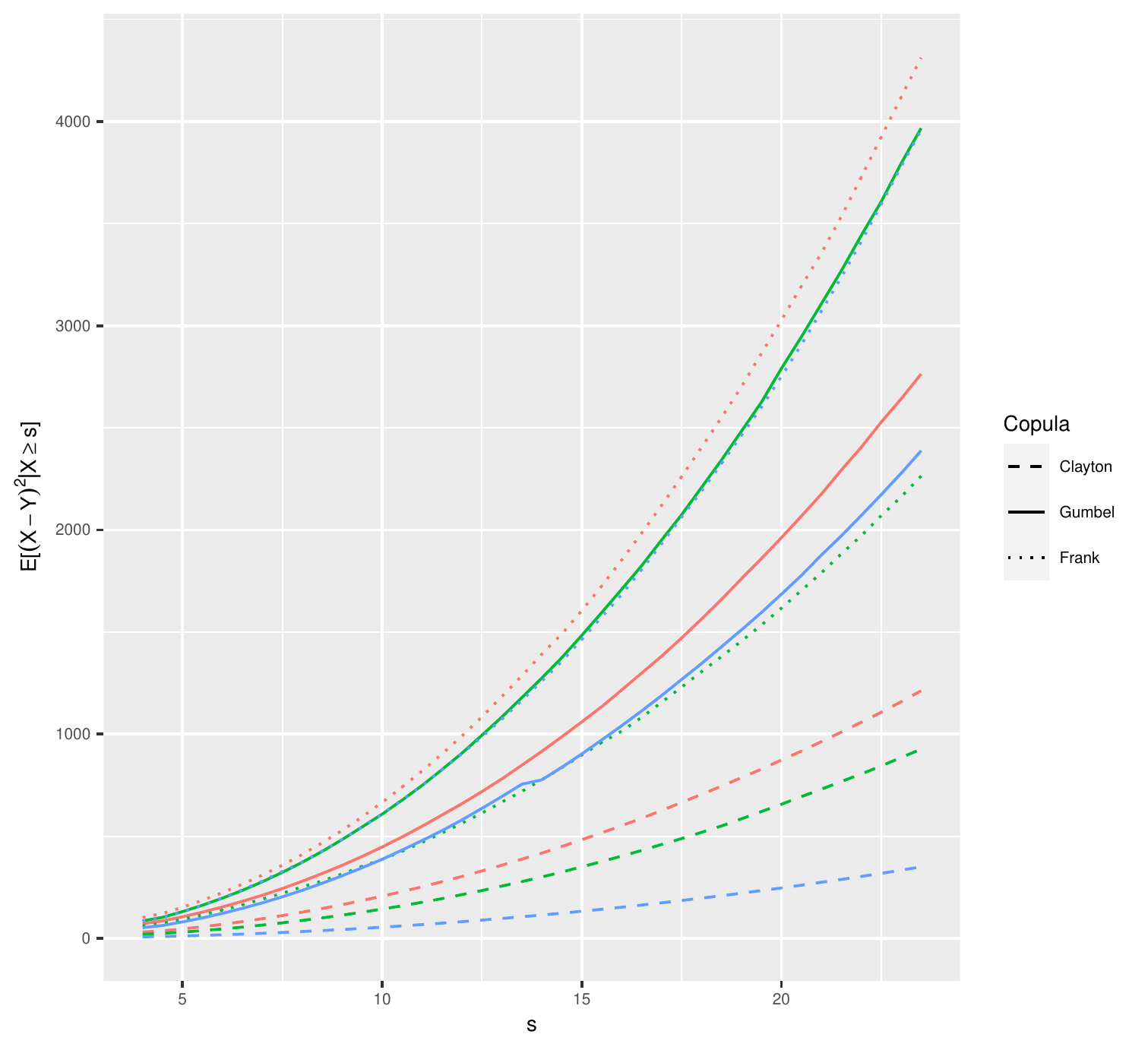} \\
     a) & b) 
     \end{tabular}
    \caption{Evolution of a) $\mathbb{E}[X-Y|X\geq s]$ and of b)  $\mathbb{E}[(X-Y)^2|X\geq s]$ with respect to the threshold $s$ for different Kendall's tau coefficients $\tau=0.3$ (red), $\tau=0.5$ (green)  $\tau=0.7$ (blue).}
    \label{fig1}
\end{figure}

% \begin{figure}[!h]
%     \centering
%     \includegraphics[scale = 1]{P1.pdf}
%     \caption{Evolution of $\mathbb{E}[X-Y|X\geq s]$ with respect to the threshold $s$ for different Kendall's tau coefficients $\tau=0.3$ (red), $\tau=0.5$ (green)  $\tau=0.7$ (blue).}
%     \label{fig1}
% \end{figure}

% \begin{figure}[!h]
%     \centering
%     \includegraphics[scale=1]{P2.pdf}
%     \caption{Evolution of $\mathbb{E}[(X-Y)^2|X\geq s]$ with respect to the threshold $s.$ The red lines correspond to a Kendall's tau coefficient $\tau=0.3,$ orange $\tau=0.5,$ blue $\tau=0.7.$}
%     \label{fig2}
% \end{figure}

We can observe that the evolution seems approximately linear for $s$ large in the case of $\mathbb{E}[(X-Y)|X\geq s],$ and approximately quadratic for $\mathbb{E}[(X-Y)^2|X\geq s]$. It is clear that the dependence structure matters, allowing to reduce the slope (which is an expected property but not covered by Proposition \ref{prop1} and \ref{prop3}). We also see in Figure \ref{fig1} a) that the slope is close to 1 for Frank's copula, and is smaller for the families that allow tail dependence. Note that a slope close to 1 is bad news: this means that the difference between $Y$ and $X$ is of the same magnitude as $X$ itself (since $X\geq s$).

These results highlight the need to use a parameter that is not only related to the basis risk via a form of ``central dependence'', but can also incorporate tail dependence. Without this property, heavy tailed variables may not be approximated properly due to the large proportion of events in the tail of the distribution.

\subsection{Impact of the heaviness of the tail}

To better understand the impact of the heaviness of the tail of the distributions and the impact of the difference between the values of the tail indices of $X$ and $Y$, we consider three different settings, denoted $B_1$, $B_2$ and $B_3$, described below.

\begin{itemize}
    \item Setting $B_1:$ $Y$ is simulated as in the main settings, but $X=Y+\varepsilon,$ with $\varepsilon\sim\mathcal{N}(0,\sigma^2_1).$ The variance $\sigma^2_1$ is taken so that $X$ has the same variance as in the main settings.
    \item Setting $B_2:$ $Y$ is simulated as in the main settings, but $\log X=\log Y+\varepsilon,$ where $\varepsilon\sim\mathcal{N}(0,\sigma^2_2).$ Again, $\sigma^2_2$ is taken so that $X$ has the same variance as in the main settings.
    \item Setting $B_3:$ $(X,Y)$ is a Gaussian vector as in (\ref{loi_gauss}), with same mean and variance as in the main settings. We consider different values for the correlation coefficient, $\rho=0.3,$ $\rho=0.5$ and $\rho=0.7.$ 
\end{itemize}

All of these benchmark cases can be thought has ``favorable cases": in $B_1$ and $B_2,$ the tail indices of $X$ and $Y$ are the same. In the first situation, we consider an additive Gaussian error, so that $X$ is relatively concentrated around $Y.$ In $B_2,$ the errors are multiplicative, since the Gaussian error is applied to the logarithms. This typically corresponds to the optimistic case where 
$\mathbb{E}[\log X|Y]=\alpha+\beta Y$: that is a standard linear regression model on the logarithm of $X$ with no misspecification error. Finally, the benchmark $B_3$ is the most optimistic case, where $X$ and $Y$ have Gaussian tails.
Next, Figure \ref{fig2} shows comparisons of $\mathbb{E}[(X-Y)\mid X \geq s]$ and $\mathbb{E}[(X-Y)^2\mid X \geq s]$ between the Clayton case and the benchmark settings. The conclusions for the other settings being similar, we postpone the figures to the appendix section (section \ref{sec_ap_b}).

\begin{figure}
    \centering
    \begin{tabular}{cc}
     \includegraphics[scale = 0.45]{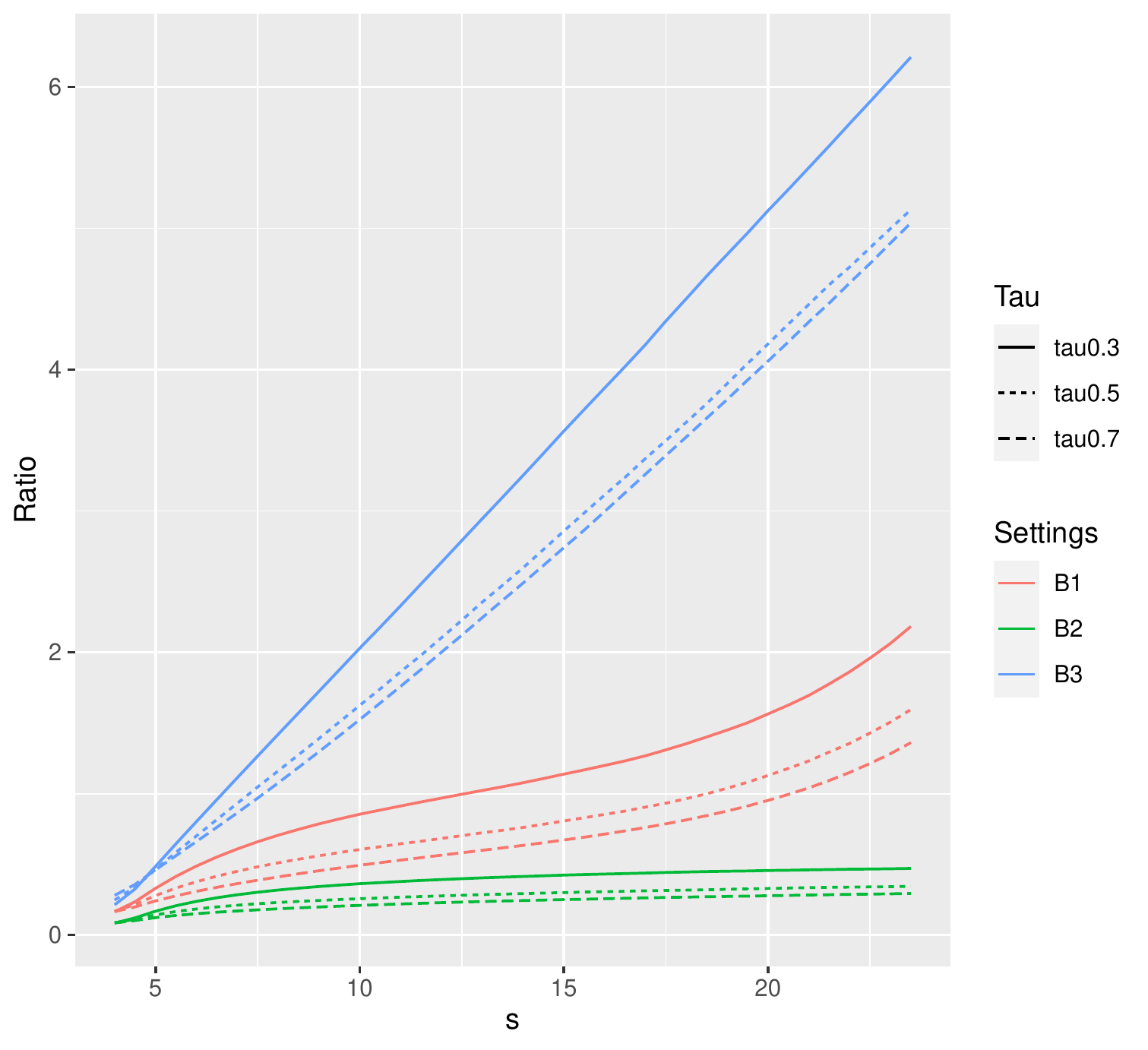} &  \includegraphics[scale = 0.45]{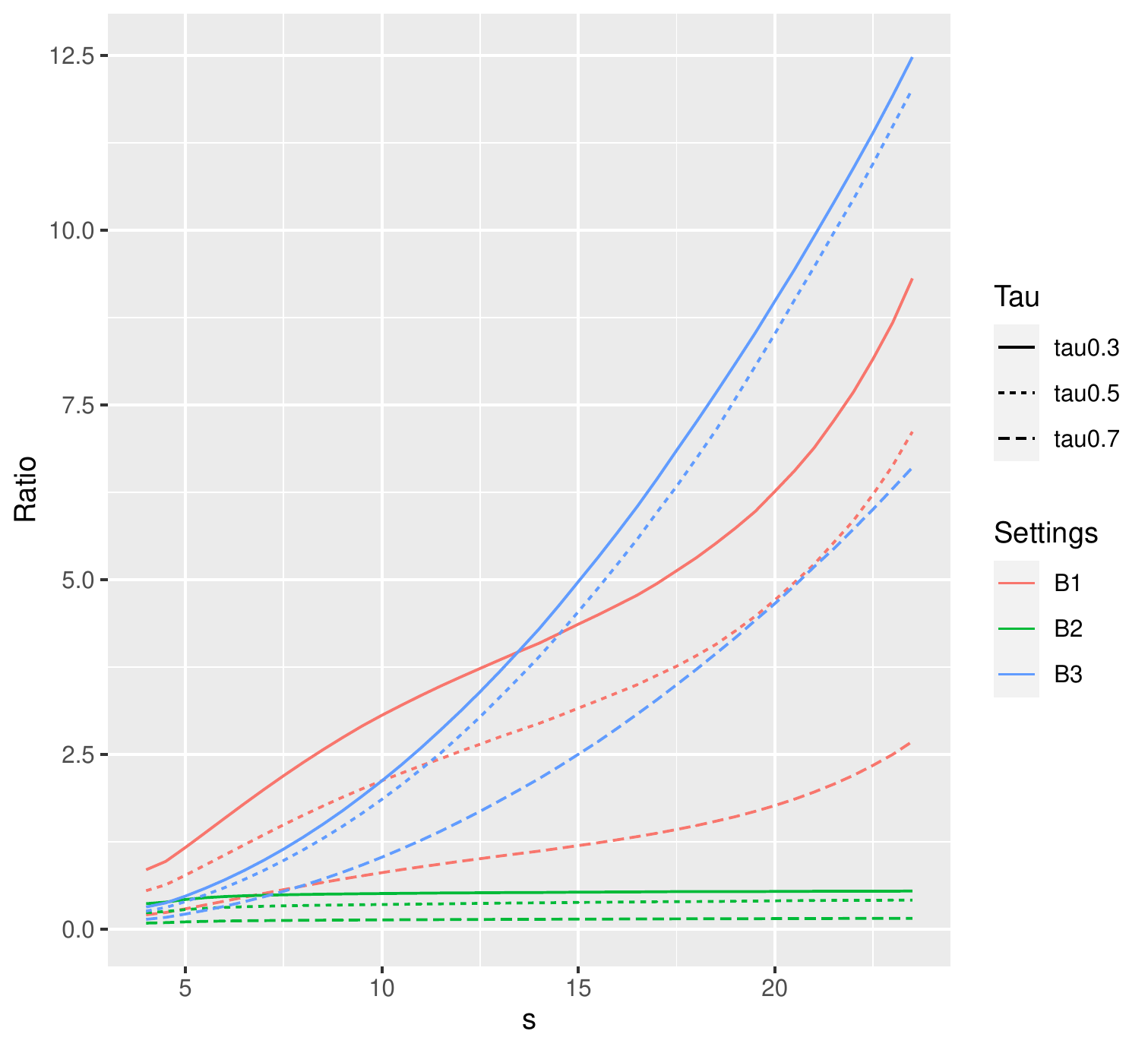} \\
     a) & b) 
     \end{tabular}
    \caption{Evolution of the ratio of a) $\mathbb{E}[(X-Y)|X\geq s]$ computed from the Clayton survival copula model, with respect to the value of $\mathbb{E}[(X-Y)|X\geq s]$  and  b) $\mathbb{E}[(X-Y)^2|X\geq s]$ computed from the Clayton survival copula model, with respect to the value of $\mathbb{E}[(X-Y)^2|X\geq s]$ obtained in the benchmark settings.}
    \label{fig2}
\end{figure}

From these figures, we see that the cases $B_1$ and $B_2$ are much favorable. In $B_3,$ we are in the Gaussian case, that is $X$ and $Y$ have low light tails and $Y$ should have the same tail, but we see that, although these cases seem very optimistic, the absence of tail dependence make the results poorer in some situations.

\section{Real data example: extreme floods} \label{sec:floods}

In this section, we illustrate our theoretical results on a real dataset, in particular, the problem created by extreme losses. The dataset comes from the global analysis of \citet{owidnaturaldisasters} on natural disasters and in particular on the share of deaths attributable to natural disasters. In this work, we chose to focus only on floods: the dataset contains 1~133 flooding events that occurred between 1950 and 2020 in 142 different countries. The dataset provides the total economic damage of each event, the country, the year and the numbers of affected people. The economic damages range from \$3 to \$40~317~000, with an average economic damage of \$788~750 and a median economic damage of \$52~403. This significant difference between the mean and the median already suggests that the economic damage variable should be heavy-tailed. 

To account for inflation, we performed a linear regression on the logarithm of the yearly mean of the economic damage against the year. The economic damages were then multiplied by $\exp(- a - b \text{Year} )$ where $a$ and $b$ are coefficients obtained in the linear regression. From now on, when we refer to economic damages, we mean deflated economic damages. 

As we wish to illustrate that parametric insurance in the context of extreme claims, we checked that the economic damage variable is indeed heavy-tailed. For that purpose, we used the so-called ``Peaks-over-Treshold'' method estimation of the tail index $\gamma$. The main idea is to choose a high threshold $u$ and fit a generalized Pareto distribution on the excesses above that threshold $u$. The estimation of the tail index $\gamma$ may be done by likelihood maximum. The choice of the threshold $u$ can be understood as a compromise between bias and variance: the smaller the threshold, the less valid the asymptotic approximation, leading to bias; on the other hand, a too high threshold will generate few excesses to fit the model, leading to high variance. The existing methods are mostly graphical, up to our knowledge, no straight-forward data-driven selection procedure is available \citep[see][for more details]{coles2001introduction}. Here, the threshold $u$ was chosen at the 0.8-quantile of the damages, equal to \$917. The tail index was estimated to be equal to 0.71. Figure \ref{fig:qqplot_expo} represents the quantile-quantile plot showing the empirical quantiles of the excesses of the economic damages above the chosen threshold against the quantiles of an exponential distribution. The fact that in the upper tail, the points are above the line $y=x$ confirms the heaviness of the tail of the distribution. 

\begin{figure}
    \centering
    \includegraphics[width = 0.5\linewidth]{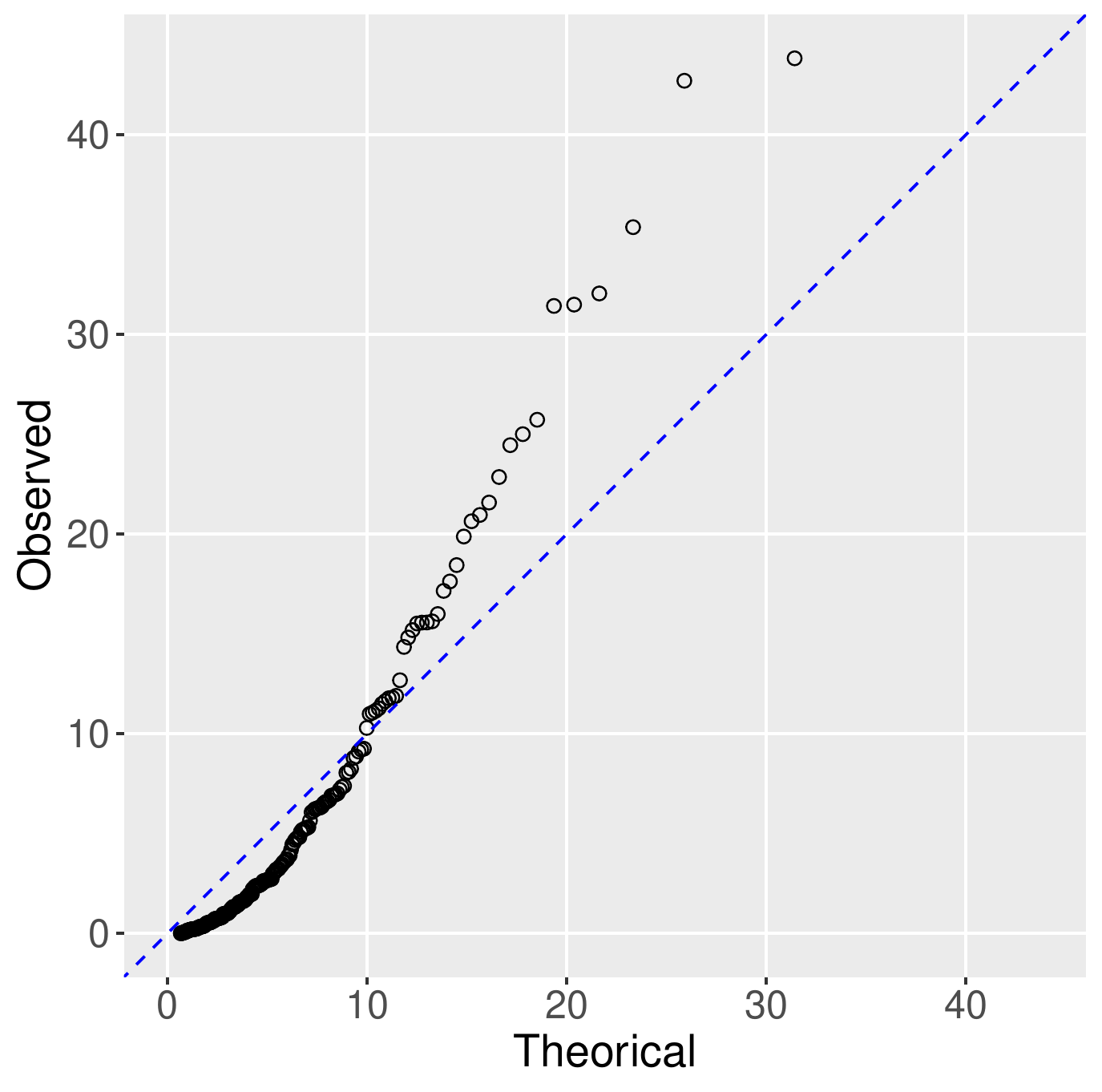}
    \caption{Quantile-quantile plot of the excesses of the economic damages against the quantiles of an exponential distribution}
    \label{fig:qqplot_expo}
\end{figure}

The actual loss $Y$ corresponds to the economic damages. The total number of affected people in a given country plays the role of the parameter $\theta$. We need to propose a function $\phi$ that links $Y$ and $\theta$, we choose the use a Classification and Regression Tree (CART) which will perform a regression of $Y$ against the number of affected people and the country in order to predict an average economic damage given the number of affected people for a given country \citep[see][for details on the CART algorithm]{breiman1984cart}. 

To illustrate the theoretical results of this paper, we performed 10-fold cross-validation, that is the data were partitioned into 10 subsamples. Of the 10 subsamples, one is kept aside as a test set and the model is trained on the 9 remaining subsamples. The cross-validation process is then repeated 10 times, with each of the 10 subsamples used exactly once as the validation data. The validation is done using the quadratic error for each prediction, that is the squared difference between the real economic damage and the predicted economic damage. At the end of the whole cross-validation procedure, classes of the real economic damages are formed (the classes are formed so that there is 10\% of the data in each class). For each class, the square root of the mean of the quadratic errors (RMSE) are computed. Figure $\ref{fig:rmse}$ shows the boxplots of the RMSE for each cost class, illustrating the fact that prediction of the cost using parametric insurance is challenging for severe costs.

\begin{figure}
    \centering
    \includegraphics[width = 0.8\linewidth]{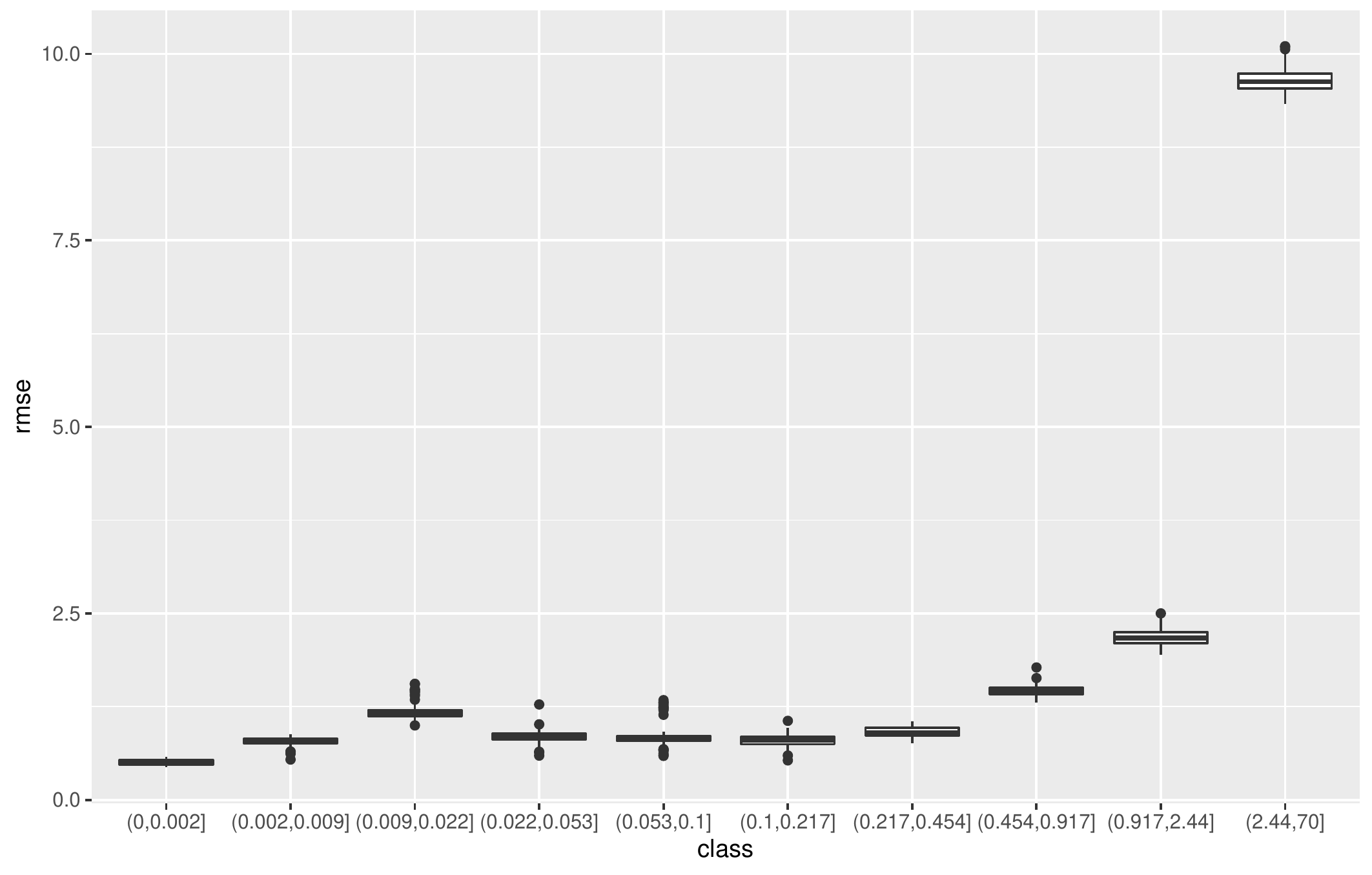}
    \caption{Boxplots of the root of the  mean squared error within each cost class for the floods data}
    \label{fig:rmse}
\end{figure}

\section{Conclusion}

In this paper, we have investigated what factors can disturb the parametric insurance project to properly cover the basis risk in case of large claims. This question can be seen of the ability of one random variable to approximate another. In particular, we have focused on the issue created by heavy-tailed losses. In these cases, the difference between what is paid to the policyholder and the actual loss may be quite large, especially if the tail index is even slightly mis-estimated. The ability of parametric insurance to reduce the variance will be diminished in this case---unless one agrees to provide poorer coverage for large claims---since no reduction of the tail index compared to the basis risk leads to, more or less, the same kind of variability. Next, the dependence structure between the parameter and the true loss seems to be an important issue to address. Tail dependence appears to be essential in order to obtain a correct approximation of the losses for large claims. It should be mentioned that designing a tail-dependent parameter for the basis risk is a challenging task: the analysis of ``extreme" events requires lots of data, which pleads for a careful statistical analysis to properly define the appropriate metric. Alternatively, a specific treatment for large claims could be considered, in order not to create too large a gap between the expectations of the policyholders and the amount compensated . Such a mixed approach---parametric coverage complemented by special treatment of large claims---could be essential to successfully deploy such strategies when facing heavy tail risk.

\section{Appendix}
\label{sec_app}

The Appendix section is organized in the following way. Section \ref{sec_ap_gauss} is devoted to the results of Gaussian variables, while the results regarding variables with Pareto tail are gathered in Section \ref{sec_ap_p}. The additional comparisons with the benchmarks of the simulation study are shown in section \ref{sec_ap_b}.

\subsection{Results on Gaussian variables}
\label{sec_ap_gauss}

\subsubsection{Proof of Proposition \ref{prop_gauss1}}

First recall that $X$ is distributed according to the distribution $\mathcal N(\mu_X, \sigma^2_X)$ so that 
\begin{equation}
\label{mrlg}
\mathbb{E}[X|X\geq s]=\mu_X+\sigma_X^2 h(s|\mu_X,\sigma_X^2),
\end{equation}
where $h(s|\mu_X,\sigma_X^2)$ is the hazard rate of a Gaussian random variable with mean $\mu_X$ and variance $\sigma_X^2,$ that is
$$h(s|\mu_X,\sigma_X^2)=\frac{\exp\left(-\frac{(s-\mu_X)^2}{2\sigma_X^2}\right)}{\sqrt{2\pi \sigma^2_X}\bar{\Phi}\left(\frac{s-\mu_X}{\sigma_X}\right)}\underset{s \to +\infty}{\sim} \frac{s-\mu_X}{\sigma_X^2},$$
with $\bar \Phi$ the survival function of the standard Gaussian distribution $\mathcal N(0,1)$.

 First of all, let
$m(s)=\mathbb{E}[(X-Y)\mathbf{1}_{X\geq s}].$ 
Since $$\mathbb{E}[Y|X]=\mu_Y+\frac{\rho \sigma_Y}{\sigma_X}(X-\mu_X),$$
we have
$$m(s)=\left(1-\frac{\rho \sigma_Y}{\sigma_X}\right)E\left[X\mathbf{1}_{X\geq s}\right]-\left(\mu_Y-\frac{\rho \sigma_Y}{\sigma_X}\mu_X\right)\mathbb{P}(X\geq s).$$

From (\ref{mrlg}), we get
\begin{eqnarray*}
\mathbb{E}[X-Y|X\geq s] &=& (\mu_X-\mu_Y)+\left(1-\frac{\rho \sigma_Y}{\sigma_X}\right)\sigma_X^2 h(s|\mu_X,\sigma_X^2) \\
&\sim & (\mu_X-\mu_Y)+\left(1-\frac{\rho \sigma_Y}{\sigma_X}\right)(s-\mu_X),
\end{eqnarray*}
as $s \to \infty$.

\subsubsection{Proof of Proposition \ref{prop_gauss2}}
\label{sec_preuve_gauss2}

We have
$$\mathbb{E}[(X-Y)^2|X\geq s]=\mathbb{E}[X^2|X\geq s]+\mathbb{E}[Y^2|Y\geq s]-2\mathbb{E}[XY|X\geq s].$$
Moreover,
$$m_2(s)=\mathbb{E}[X^2|X\geq s]=\left(\sigma_X^2+\mu_X^2\right)+\sigma_X^2 h(s|\mu_X,\sigma_X^2)(s+\mu_X),$$
and
$$\mathbb{E}[XY|X\geq s]=\mathbb{E}[X\mathbb{E}[Y|X]|X\geq s].$$
From (\ref{mrlg}),
$$\mathbb{E}[XY|X\geq s]=\left(\mu_Y-\rho \frac{\sigma_Y}{\sigma_X}\right)(\mu_X+\sigma_X^2h(s|\mu_X,\sigma_X^2))+\rho\frac{\sigma_Y}{\sigma_X}m_2(s).$$
Finally,
$$\mathbb{E}[Y^2|X\geq s]=\mathbb{E}[\mathrm{Var}(Y|X)|X\geq s]+\mathbb{E}[\mathbb{E}[Y|X]^2|X\geq s],$$
which leads to
\begin{eqnarray*}\mathbb{E}[Y^2|X\geq s] &=& \sigma_Y^2(1-\rho^2)+\left(\mu_Y-\rho\mu_X\frac{\sigma_Y}{\sigma_X}\right)^2 \\
&&+\rho^2\frac{\sigma_Y^2}{\sigma_X^2}m_2(s)+2\rho \frac{\sigma_Y}{\sigma_X}(\mu_Y-\rho\mu_X\frac{\sigma_Y}{\sigma_X})(\mu_X+\sigma_X^2h(s|\mu_X,\sigma_X^2)).
\end{eqnarray*}
Hence, we see that
\begin{eqnarray*}
\mathbb{E}[(X-Y)^2|X\geq s] &=& \left(1-\rho \frac{\sigma_Y}{\sigma_X}\right)^2m_2(s)+o(m_2(s)) \\
&=& \left(1-\rho \frac{\sigma_Y}{\sigma_X}\right)^2sh(s|\mu_X,\sigma_X^2) +o(sh(s|\mu_X,\sigma_X^2)) \\
&\sim& \left(1-\rho \frac{\sigma_Y}{\sigma_X}\right)^2\frac{s^2}{\sigma^2_X},
\end{eqnarray*}
 as $s \to \infty$.
\subsection{Results on heavy-tail variables}
\label{sec_ap_p}

\subsubsection{A preliminary result}

We start with a preliminary result showing that the variable $X-Y$ has the same tail index as $X,$ under the assumptions of Propositions \ref{prop1} to \ref{prop3}. Lemma \ref{l1} consists in providing upper and lower bounds for the survival function of $X-Y$. 

\begin{lemma}
\label{l1}
Let $X,$ $Y$ be as defined in Proposition \ref{prop1}, and let $Z=X-Y$.
Then, introducing $S_Z(t)=\mathbb{P}(Z\geq t),$ for $t\geq 0,$
$$h^{-1/\gamma_X}t^{-1/\gamma_X}\ell_X(th)-(h-1)^{-1/\gamma_Y}t^{-1/\gamma_Y}\ell_Y(t(h-1))\leq S_Z(t)\leq t^{-1/\gamma_X}\ell_X(t),$$
with $h>0$ fixed.
\end{lemma}

\begin{proof}
Since $Y\geq 0$ almost surely, $X-Y\geq t$ implies that $X\geq t.$ Hence, we get $S_Z(t)\leq S_X(t),$ and the upper bound of Lemma \ref{l1} is obtained.

To obtain the lower bound, introduce the event, for $h>0$ fixed,
$$E_{h}(t)=\left\{ \{X\geq th\} \cap \{Y\leq t(h-1)\}\right\}.$$
We have
$S_Z(t)\geq \mathbb{P}(E_h(t)).$ Next,
\begin{eqnarray*}\mathbb{P}(E_h(t)) &\geq & \mathbb{P}(X\geq th)-\mathbb{P}(Y\geq t(h-1))=
S_X(th)-S_Y(t(h-1)).
\end{eqnarray*}
This shows the lower bound in Lemma \ref{l1}.
\end{proof}

We can now apply this lemma to obtain the following proposition.

\begin{proposition}
\label{prop_tail}
Let $X$ be a heavy tail variable with $\gamma_X>0.$ Let $Y$ be a non negative variable with tail index $\gamma_Y<\gamma_X.$ If $\gamma_X>\gamma_Y,$ $Z=X-Y$ is heavy-tailed with tail index $\gamma_X.$
\end{proposition}

\begin{proof}

Let $\ell_Z(t)=t^{1/\gamma_X}S_Z(t).$ Let us proof that $\ell_Z$ is slowly-varying. From Lemma \ref{l1}, for $x>1$ and for all $h>10$, 
\begin{eqnarray}
\frac{\ell_Z(tx)}{\ell_Z(t)}  
%&=& \frac{x^{1/\gamma_X}S_Z(tx)}{S_Z(t)}\\
% &\leq& \frac{t^{-1\gamma_X}\ell_X(tx)}{h^{-1/\gamma_X}t^{-1/\gamma_X}\ell_X(th)-(h-1)^{-1/\gamma_Y}(tx)^{-1/\gamma_Y}\ell_Y(t(h-1))}\\
&\le& \frac{\ell_X(tx)}{h^{-1/\gamma_X}\ell_X(th)-(h-1)^{-1/\gamma_X}(t(h-1))^{1/\gamma_X-1/\gamma_Y}\ell_Y(t(h-1))} \, , \label{force1}
\end{eqnarray}

and,
\begin{equation} \frac{\ell_X(thx)}{h^{1/\gamma_X}\ell_X(t)}-\frac{x^{1/\gamma_X}(tx(h-1))^{-1/\gamma_Y}\ell_Y(tx(h-1))}{t^{-1/\gamma_X}\ell_X(t)} \leq \frac{\ell_Z( tx)}{\ell_Z(t)}.
\label{force2}
\end{equation}

Hence, for all $x>0$ and $h>1,$ taking the limit of the right-hand side of (\ref{force1}),
$$\lim \sup_{t\rightarrow \infty} \frac{\ell_Z(tx)}{\ell_Z(t)}\leq h^{1/\gamma_X}.$$
To see that, let $\beta=\gamma_Y^{-1}-\gamma_X^{-1}>0.$ We have
$t^{-\beta/2}l_Y(t(h-1))\rightarrow_{t\rightarrow\infty} 0,$ and $t^{\beta/2}l_X(tx)\rightarrow_{t\rightarrow\infty}\infty.$ 

Similarly, from (\ref{force2}), we get
$$\lim \inf_{t\rightarrow \infty} \frac{\ell_Z(tx)}{\ell_Z(t)}\geq \frac{1}{h^{1/\gamma_X}}.$$
This is valid for all $h>1.$ Next, let $h$ tend to 1 in order to obtain that, for all $x,$ $\lim_{t\rightarrow \infty} \ell_Z(tx)/\ell_Z(t)=1,$ leading to the result.

\end{proof}

\subsubsection{Proof of Proposition \ref{prop1}}

From Proposition \ref{prop_tail}, $Z=X-Y$ is heavy-tailed with tail index $\gamma_X$, that is $S_Z(t)=\ell_Z(t)t^{-1/\gamma_X}$ with $\ell_Z$ slow-varying. Hence,
\begin{equation}
\label{mrl}
    \mathbb{E}[Z|Z\geq s]=s+o(s).
\end{equation}

Next,
\begin{eqnarray*}\mathbb{E}[Z|X\geq s] &=& \mathbb{E}[Z|X\geq s, Z\geq s]\mathbb{P}(Z\geq s|X\geq s) \\&& + \mathbb{E}[Z|X\geq s,Z<s]\mathbb{P}(Z< s|X\geq s).
\end{eqnarray*}
Since $Z\geq s$ implies $X\geq s,$ we have
$$\mathbb{E}[Z|X\geq s, Z\geq s]\mathbb{P}(Z\geq s|X\geq s)=\mathbb{E}[Z|Z\geq s]\mathbb{P}(Z\geq s|X\geq s).$$

Moreover,
$$\mathbb{P}(Z\geq s|X\geq s)=\frac{\mathbb{P}\left(\{Z\geq s\}\cap \{X\geq s\}\right)}{\mathbb{P}(X\geq s)}=\frac{\mathbb{P}(Z\geq s)}{\mathbb{P}(X\geq s)}.$$

From this, we get
$$\mathbb{E}[Z|X\geq s]\geq \mathbb{E}[Z|Z\geq s]\times \frac{\ell_Z(s)}{\ell_X(s)}=(s+o(s))\times \frac{\ell_Z(s)}{\ell_X(s)},$$
from (\ref{mrl}).
From Lemma \ref{l1}, we see that, taking for example $x=2,$
$$\frac{\ell_Z(s)}{\ell_X(s)}\geq \frac{1}{2^{1/\gamma_X}}\frac{\ell_X(2s)}{\ell_X(s)}-\frac{1}{s^{\beta}}\frac{\ell_Y(s)}{\ell_X(s)},$$ introducing $\beta=\gamma_Y^{-1}-\gamma_X^{-1}>0.$ Since $\ell_X$ and $\ell_Y$ are slow varying, $s^{\beta/2}\ell_X(s)\rightarrow \infty,$ $s^{-\beta/2}\ell_Y(s)\rightarrow 0$ as $s$ tends to infinity, leading to
$$\frac{1}{s^{\beta}}\frac{\ell_Y(s)}{\ell_X(s)}=o(1).$$ Moreover,
$$\lim_{s \rightarrow \infty}\frac{\ell_X(2s)}{\ell_X(s)}=1,$$ showing that there exists a constant $c>0$ such that, for $s$ large enough,
$$\mathbb{E}[Z|X\geq s]\geq c s.$$

\subsubsection{Proof of Proposition \ref{prop2}}

We have
$$\mathbb{E}[X-Y|X\geq s]=\mathbb{E}[X-\psi(X)|X\geq s],$$
where $\psi(X)=\mathbb{E}[Y|X].$
Since $X$ is heavy-tailed with tail index $\gamma_X>0,$ $\mathbb{E}[X|X\geq s]=s+o(s).$
On the other hand,
$$\mathbb{E}[\psi(X)|X\geq s]=\mathbb{E}[\psi(X)|\psi(X)\geq \psi(s)],$$
since $\psi$ is strictly non decreasing. Then, since $\psi(X)$ assumed to be heavy-tailed,
% Moreover,
% $$\mathbb{P}(\psi(X)\geq t)=\mathbb{P}(X\geq \psi^{-1}(t))=\frac{\ell_X(\psi^{-1}(t))}{\psi^{-1}(t)^{1/\gamma_X}}.$$
% Since $\psi^{-1}(t){\color{red}  =} \ell_{\psi}(t)t^{-1/\gamma_{\psi}},$ then {\color{red} $\psi(X)$ is heavy-tailed (not sure !!!!)  Ok if we assume $\psi^{-1}(t)= \ell_{\psi}(t)t^{1/\gamma_{\psi}}$} and
$$\mathbb{E}[\psi(X)|\psi(X)\geq \psi(s)]=\psi(s)+o(\psi(s)),$$
which shows that
$$\mathbb{E}[X-Y|X\geq s]=s-\psi(s)+o(s).$$

\subsubsection{Proof of Proposition \ref{prop3}}

Let $Z=X-Y.$ We have
$$\mathbb{E}[(X-Y)^2|X\geq s]=\frac{\mathbb{E}[Z^2\mathbf{1}_{X\geq s}]}{\mathbb{P}(X\geq s)}\geq \frac{\mathbb{E}[Z^2\mathbf{1}_{X\geq s}\mathbf{1}_{Z\geq s}]}{\mathbb{P}(X\geq s)}.$$
If $Z\geq s$, then, necessarily, $X\geq s$ since $Y\geq 0.$ Hence
$$\mathbb{E}[(X-Y)^2|X\geq s]\geq \frac{\mathbb{E}[Z^2\mathbf{1}_{Z\geq s}]}{\mathbb{P}(X\geq s)}=\mathbb{E}[Z^2|Z\geq s]\frac{\mathbb{P}(Z\geq s)}{\mathbb{P}(X\geq s)}.$$

Next,
$$\mathbb{E}[Z^2|Z\geq s]=\frac{\mathbb{E}[Z^2\mathbf{1}_{Z^2\geq s^2}]}{\mathbb{P}(Z\geq s)}-\frac{\mathbb{E}[Z^2\mathbf{1}_{Z<-s}]}{\mathbb{P}(Z\geq s)}.$$
Moreover,
$$\mathbb{E}[Z^2\mathbf{1}_{Z<-s}]\leq \mathbb{E}[Y^2\mathbf{1}_{Y\geq s}]=\mathbb{E}[Y^2|Y\geq s]\mathbb{P}(Y\geq s).$$
Combining this last equation with Proposition \ref{prop_tail} leads to
\begin{eqnarray*}
\mathbb{E}[Z^2|Z\geq s] &\geq & \mathbb{E}[Z^2|Z^2\geq s^2]\frac{\mathbb{P}(Z^2\geq s^2)}{\mathbb{P}(Z\geq s)}-\frac{\ell_Y(s)}{\ell_Z(s)}\mathbb{E}[Y^2|Y\geq s]s^{1/\gamma_X-1/\gamma_Y} \\
&\geq &\mathbb{E}[Z^2|Z^2\geq s^2]-\frac{\ell_Y(s)}{\ell_Z(s)}\mathbb{E}[Y^2|Y^2\geq s^2]s^{1/\gamma_X-1/\gamma_Y},
\end{eqnarray*}
where the last line comes from the fact that, $\mathbb{P}(Z^2\geq s^2)\geq \mathbb{P}(Z\geq s),$ and that, since $Y\geq 0$ almost surely, $\mathbb{E}[Y^2|Y\geq s]=\mathbb{E}[Y^2|Y^2\geq s^2].$

We have, since $Y\geq 0$ almost surely,
$$\mathbb{P}(Y^2\geq t)=\mathbb{P}(Y\geq t^{1/2})=\frac{\ell_Y(t^{1/2})}{t^{1/(2\gamma_Y)}},$$
where $t\rightarrow \ell_Y(t^{1/2})$ inherits the slow-varying property of $\ell_Y.$ Hence 
\begin{equation}\label{ey2}
\mathbb{E}[Y^2|Y^2\geq s^2]=s^2+o(s^2).
\end{equation}

On the other hand, let $\beta=\gamma_Y^{-1}-\gamma_X^{-1}>0.$ Since $\ell_Y(s)s^{2-\beta/2}\rightarrow 0$ when $s$ tends to infinity, and since $\ell_Z(s)s^{\beta/2}\rightarrow \infty.$ Hence, using (\ref{ey2}),
$$\frac{\ell_Y(s)}{\ell_Z(s)}\mathbb{E}[Y^2|Y^2\geq s^2]s^{1/\gamma_X-1/\gamma_Y}= s^{2-\beta/2}+o(s^{2-\beta/2}).$$

Since $\mathbb{E}[Z^2|Z^2\geq s^2] = s^2+o(s^2),$ the result follows.

\newpage

\subsection{Additional comparisons with benchmarks}
\label{sec_ap_b}

\begin{figure}
    \centering
    \begin{tabular}{cc}
     \includegraphics[scale = 0.5]{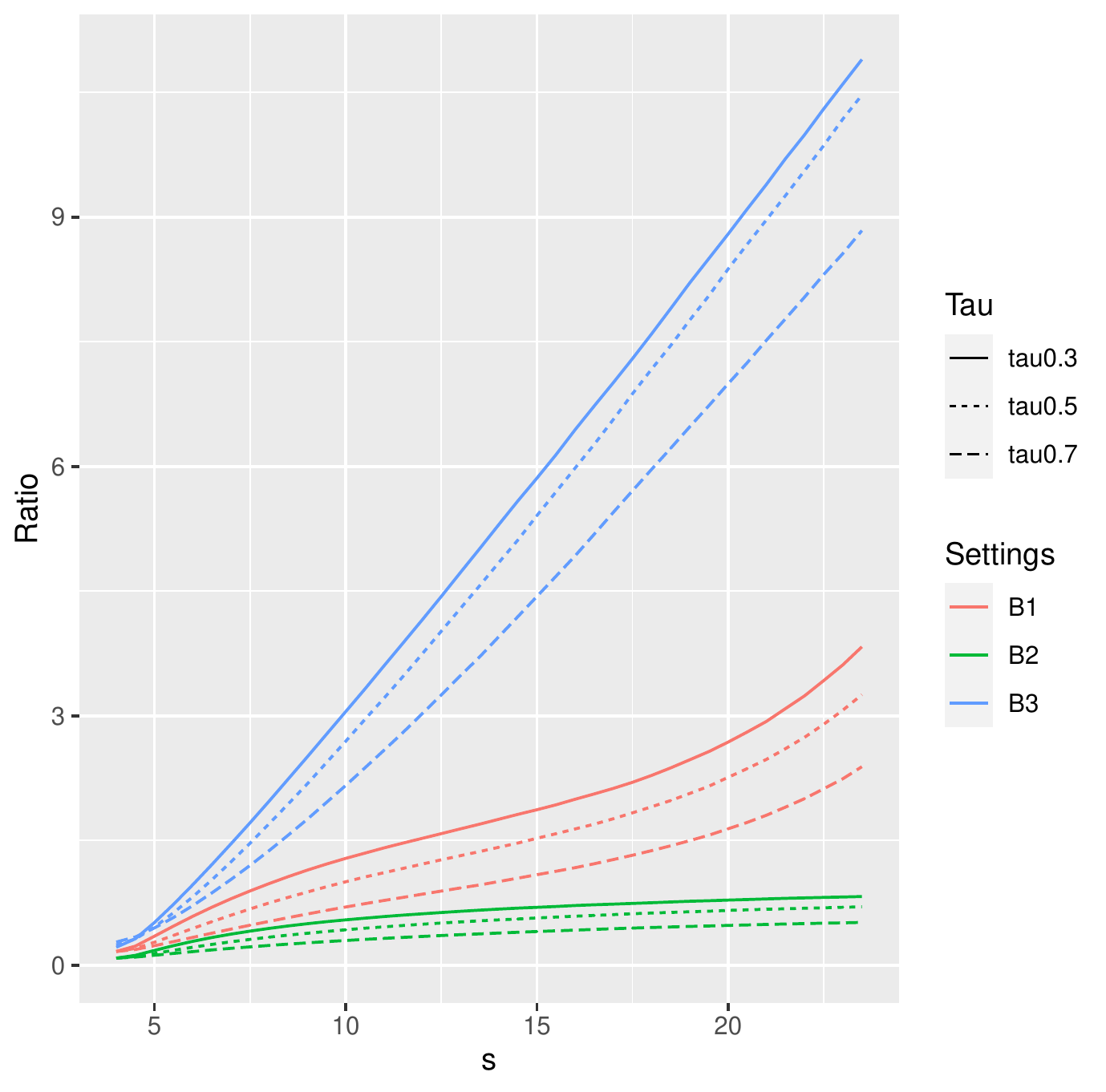} &  \includegraphics[scale = 0.5]{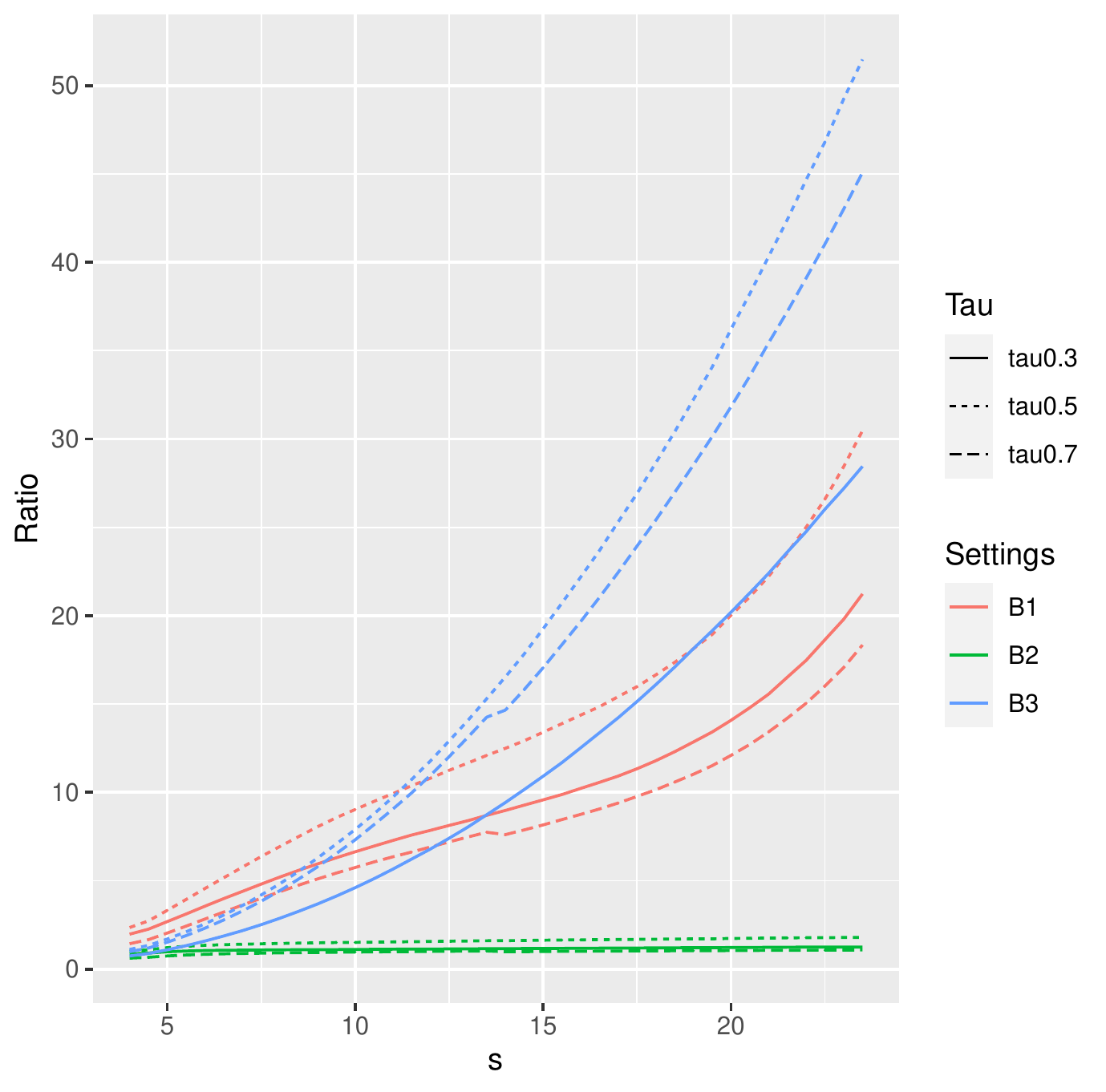} \\
     a) & b) 
     \end{tabular}
  \caption{Evolution of the ratio of a) $\mathbb{E}[(X-Y)|X\geq s]$ computed from the Clayton survival copula model, with respect to the value of $\mathbb{E}[(X-Y)|X\geq s]$  and  b) $\mathbb{E}[(X-Y)^2|X\geq s]$ computed from the Gumbel copula model, with respect to the value of $\mathbb{E}[(X-Y)^2|X\geq s]$ obtained in the benchmark settings.}
    \label{fig3}
\end{figure}

\begin{figure}
    \centering
    \begin{tabular}{cc}
     \includegraphics[scale = 0.5]{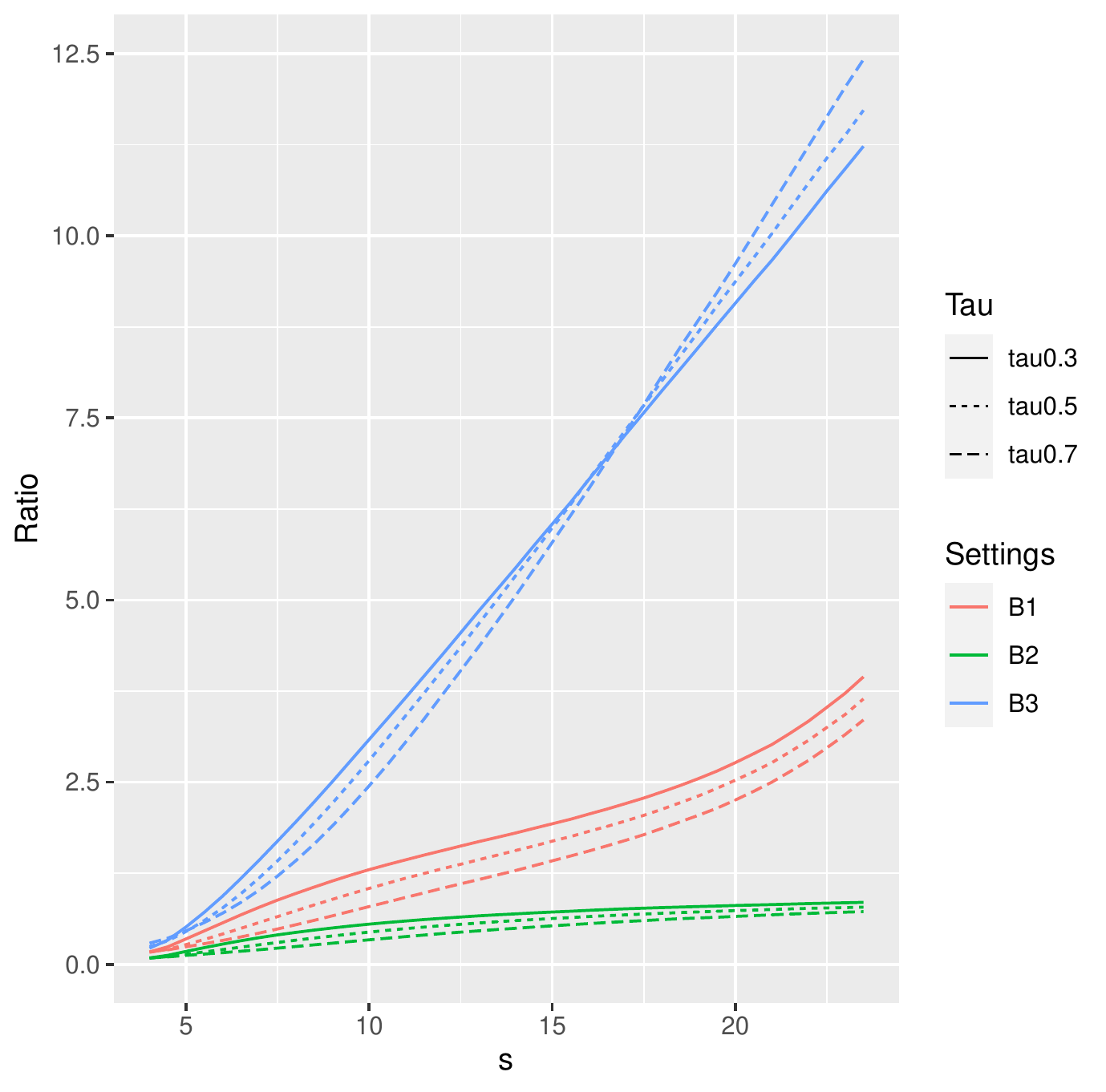} &  \includegraphics[scale = 0.5]{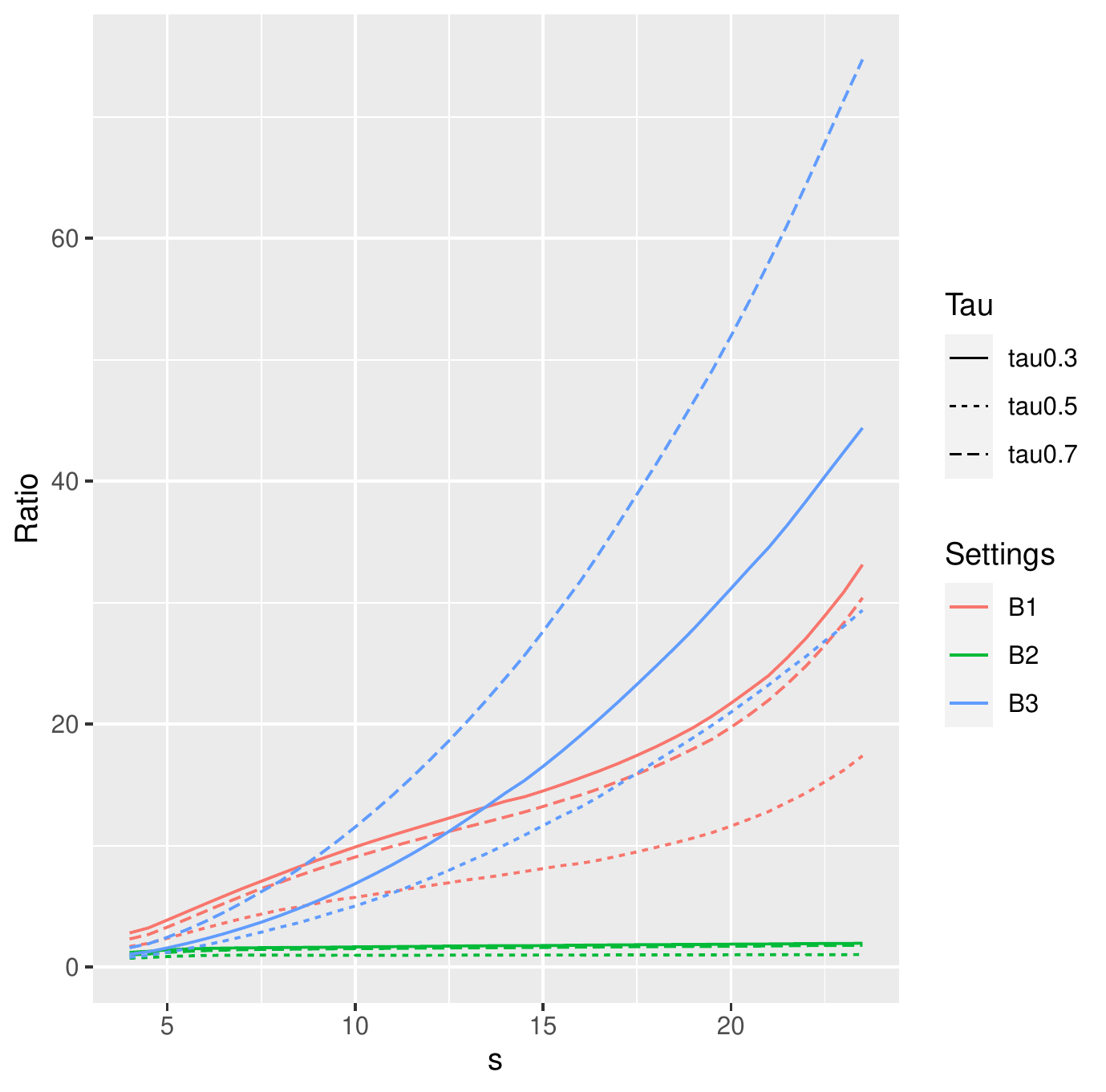} \\
     a) & b) 
     \end{tabular}
  \caption{Evolution of the ratio of a) $\mathbb{E}[(X-Y)|X\geq s]$ computed from the Clayton survival copula model, with respect to the value of $\mathbb{E}[(X-Y)|X\geq s]$  and  b) $\mathbb{E}[(X-Y)^2|X\geq s]$ computed from the Frank copula model, with respect to the value of $\mathbb{E}[(X-Y)^2|X\geq s]$ obtained in the benchmark settings.}
    \label{fig4}
\end{figure}

% \begin{figure}[!h]
%     \centering
%     \includegraphics[scale=0.10]{P7.png}
%     \caption{Evolution of the ratio of $\mathbb{E}[(X-Y)^2|X\geq s]$ computed from the Gumbel copula model, with respect to the value of $\mathbb{E}[(X-Y)|X\geq s]$ obtained in the benchmark settings. The orange lines correspond to benchmark $B_1,$ the red ones to scenario $B_2,$ the blue one to scenario $B_3.$}
%     \label{fig7}
% \end{figure}

% \begin{figure}[!h]
%     \centering
%     \includegraphics[scale=0.10]{P8.png}
%     \caption{Evolution of the ratio of $\mathbb{E}[(X-Y)^2|X\geq s]$ computed from the Frank copula model, with respect to the value of $\mathbb{E}[(X-Y)|X\geq s]$ obtained in the benchmark settings. The orange lines correspond to benchmark $B_1,$ the red ones to scenario $B_2,$ the blue one to scenario $B_3.$}
%     \label{fig8}
% \end{figure}

\textbf{Acknowledgment:} The authors acknowledge funding from the project \textit{Cyber Risk Insurance: actuarial modeling}, Joint Research Initiative under the aegis of Risk Fundation, with partnership of AXA, AXA GRM, ENSAE and Sorbonne Universit\'e.

\bibliographystyle{apalike}
\bibliography{bibli}

\begin{thebibliography}{}

\bibitem[Albrecher et~al., 2004]{albrecher2004qmc}
Albrecher, H., Hartinger, J., and Tichy, R.~F. (2004).
\newblock {QMC} techniques for {CAT} bond pricing.
\newblock {\em Monte Carlo Methods and Applications}, 10(3-4):197--211.

\bibitem[Balkema and De~Haan, 1974]{balkema1974residual}
Balkema, A.~A. and De~Haan, L. (1974).
\newblock Residual life time at great age.
\newblock {\em The Annals of probability}, 2(5):792--804.

\bibitem[Beirlant et~al., 2004]{beirlant}
Beirlant, J., Goegebeur, Y., Segers, J., and Teugels, J.~L. (2004).
\newblock {\em Statistics of extremes: theory and applications}, volume 558.
\newblock John Wiley \& Sons.

\bibitem[Bodily and Coleman, 2021]{bodily2021portfolio}
Bodily, S.~E. and Coleman, D.~M. (2021).
\newblock A portfolio of weather risk transfer contracts efficiently reduces
  risk.
\newblock {\em Climate Risk Management}, page 100332.

\bibitem[Breiman et~al., 1984]{breiman1984cart}
Breiman, L., Friedman, J., Olshen, R., and Stone, C. (1984).
\newblock Cart.
\newblock {\em Classification and Regression Trees}.

\bibitem[Broberg, 2020]{broberg2020parametric}
Broberg, M. (2020).
\newblock Parametric loss and damage insurance schemes as a means to enhance
  climate change resilience in developing countries.
\newblock {\em Climate Policy}, 20(6):693--703.

\bibitem[Coles, 2001]{coles2001introduction}
Coles, S. (2001).
\newblock {\em An introduction to statistical modeling of extreme values}.
\newblock Springer.

\bibitem[Dal~Moro, 2020]{dal2020towards}
Dal~Moro, E. (2020).
\newblock Towards an economic cyber loss index for parametric cover based on it
  security indicator: a preliminary analysis.
\newblock {\em Risks}, 8(2):45.

\bibitem[Farkas et~al., 2021]{farkas2021cyber}
Farkas, S., Lopez, O., and Thomas, M. (2021).
\newblock Cyber claim analysis using generalized pareto regression trees with
  applications to insurance.
\newblock {\em Insurance: Mathematics and Economics}, 98:92--105.

\bibitem[Figueiredo et~al., 2018]{figueiredo2018probabilistic}
Figueiredo, R., Martina, M.~L., Stephenson, D.~B., and Youngman, B.~D. (2018).
\newblock A probabilistic paradigm for the parametric insurance of natural
  hazards.
\newblock {\em Risk Analysis}, 38(11):2400--2414.

\bibitem[Gnedenko, 1943]{gnedenko}
Gnedenko, B. (1943).
\newblock Sur la distribution limite du terme maximum d'une serie aleatoire.
\newblock {\em Annals of mathematics}, pages 423--453.

\bibitem[Horton, 2018]{horton2018parametric}
Horton, J.~B. (2018).
\newblock Parametric insurance as an alternative to liability for compensating
  climate harms.
\newblock {\em Carbon \& Climate Law Review}, 12(4):285--296.

\bibitem[Jacobs, 2014]{cost_jacobs_formula}
Jacobs, J. (2014).
\newblock Analyzing {P}onemon {C}ost of {D}ata {B}reach.

\bibitem[Jarrow, 2010]{jarrow2010simple}
Jarrow, R.~A. (2010).
\newblock A simple robust model for cat bond valuation.
\newblock {\em Finance Research Letters}, 7(2):72--79.

\bibitem[Johnson, 2021]{johnson2021paying}
Johnson, L. (2021).
\newblock Paying ex gratia: Parametric insurance after calculative devices
  fail.
\newblock {\em Geoforum}, 125:120--131.

\bibitem[Lin and Kwon, 2020]{lin2020application}
Lin, X. and Kwon, W.~J. (2020).
\newblock Application of parametric insurance in principle-compliant and
  innovative ways.
\newblock {\em Risk Management and Insurance Review}, 23(2):121--150.

\bibitem[Maillart and Sornette, 2010]{maillart}
Maillart, T. and Sornette, D. (2010).
\newblock Heavy-tailed distribution of cyber-risks.
\newblock {\em The European Physical Journal B}, 75(3):357--364.

\bibitem[Mikosch and Nagaev, 1998]{mikosch}
Mikosch, T. and Nagaev, A.~V. (1998).
\newblock Large deviations of heavy-tailed sums with applications in insurance.
\newblock {\em Extremes}, 1(1):81--110.

\bibitem[Nelsen, 2007]{nelsen}
Nelsen, R.~B. (2007).
\newblock {\em An introduction to copulas}.
\newblock Springer Science \& Business Media.

\bibitem[OCDE, 2017]{occ}
OCDE (2017).
\newblock Enhancing the role of insurance in cyber risk management.

\bibitem[Ponemon, 2018]{CODB_2018}
Ponemon, L. (2018).
\newblock Cost of a data breach study: global overview.
\newblock {\em Benchmark research sponsored by IBM Security independently
  conducted by Ponemon Institute LLC}.

\bibitem[Prokopchuk et~al., 2020]{prokopchuk2020parametric}
Prokopchuk, O., Prokopchuk, I., Mentel, G., and Bilan, Y. (2020).
\newblock Parametric insurance as innovative development factor of the
  agricultural sector of economy.
\newblock {\em AGRIS on-line Papers in Economics and Informatics},
  10(665-2021-565):69--86.

\bibitem[Ritchie and Roser, 2014]{owidnaturaldisasters}
Ritchie, H. and Roser, M. (2014).
\newblock Natural disasters.
\newblock {\em Our World in Data}.
\newblock https://ourworldindata.org/natural-disasters.

\bibitem[Sklar, 1959]{sklar}
Sklar, M. (1959).
\newblock Fonctions de r\'epartition à {$n$} dimensions et leurs marges.
\newblock {\em Publ. Inst. Statist. Univ. Paris}, 8:229--231.

\bibitem[Van~Nostrand and Nevius, 2011]{van2011parametric}
Van~Nostrand, J.~M. and Nevius, J.~G. (2011).
\newblock Parametric insurance: using objective measures to address the impacts
  of natural disasters and climate change.
\newblock {\em Environmental Claims Journal}, 23(3-4):227--237.

\bibitem[Zimbidis et~al., 2007]{zimbidis2007modeling}
Zimbidis, A.~A., Frangos, N.~E., and Pantelous, A.~A. (2007).
\newblock Modeling earthquake risk via extreme value theory and pricing the
  respective catastrophe bonds.
\newblock {\em ASTIN Bulletin: The Journal of the IAA}, 37(1):163--183.

\end{thebibliography}

\end{document}